\documentclass[journal,draftcls,onecolumn]{IEEEtran}
\usepackage[utf8]{inputenc}
\usepackage{amsmath,amsthm,amssymb}
\usepackage{cite}
\usepackage{url}
\usepackage{float}

\usepackage{xcolor}
\usepackage{graphicx}
\usepackage{tikz}
\usepackage{bbm}
\usepackage{comment}
\usepackage{pgfplots}
\pgfplotsset{compat=1.18}
\usepackage{tikz}
\usetikzlibrary{patterns}

\newcommand{\xtic}[1]{(#1,0) -- ++(0,-0.25)}
\newcommand{\ytic}[1]{(0,#1) -- ++(-0.25,0)}
\newcommand{\I}[1]{I_{\set{#1}}}

\definecolor{OI-vermillion}{RGB}{213,94,0}
\usepackage[colorlinks=true,allcolors=OI-vermillion]{hyperref}
\usepackage[color=red!30,colorinlistoftodos]{todonotes}
\usepackage[normalem]{ulem}
\newcommand\redout{\bgroup\markoverwith{\textcolor{red}{\rule[.5ex]{2pt}{0.4pt}}}\ULon}

\DeclareMathOperator*{\argmin}{arg\,min}

\long\def\/*#1*/{}

\usepackage{bm}

\newcommand{\MLC}{\beta}

\renewcommand{\L}{\mathcal{L}}

\newcommand{\X}{\mathcal{X}}

\newtheorem{proposition}{Proposition}
\newtheorem{theorem}{Theorem}
\newtheorem{definition}{Definition}

\newtheorem{lemma}{Lemma}
\newtheorem{corollary}{Corollary}
\usepackage{texdef2020}

\newcommand{\spmf}[2]{\pmf{}{#2}}
\newcommand{\support}[2][n]{\mathcal{#2}_{#1}}

\newcommand{\lambar}{\bar{\lambda}}

\newcommand{\AU}{\mathcal{U}}
\newcommand{\AY}{\mathcal{Y}}

\newcommand{\TMaxL}{T_{\Lambda}}
\newcommand{\papertitle}{Balancing Timeliness and Privacy in Discrete-Time Updating Systems}

\title{\papertitle
\thanks{ 
A preliminary version of these results appeared in the Proceedings of the 2022 IEEE International Symposium on Information Theory (ISIT)~\cite{NityaPrivacy}. 
This work was supported in part by the US National Science Foundation under awards SaTC-1617849, CCF-1909468, and CNS-2148104.}
\thanks{The authors are with the Department of Electrical and Computer Engineering, Rutgers, The State University of New Jersey, 94 Brett Road, Piscataway, NJ 08854, USA (email: \href{mailto:nitya.s@rutgers.edu}{nitya.s@rutgers.edu}, \href{mailto:anand.sarwate@rutgers.edu}{anand.sarwate@rutgers.edu}, \href{narayan@winlab.rutgers.edu}{narayan@winlab.rutgers.edu}, \href{mailto:ryates@winlab.rutgers.edu}{ryates@winlab.rutgers.edu}).}
}

\author{
    Nitya Sathyavageeswaran,~\IEEEmembership{Student Member,~IEEE,}
\and
    Anand D.~Sarwate,~\IEEEmembership{Senior Member,~IEEE,}
\and
    Narayan B.~Mandayam,~\IEEEmembership{Fellow,~IEEE,} 
\and
    Roy D.~Yates,~\IEEEmembership{Fellow,~IEEE} 
}

\IEEEpubid{0000--0000/00\$00.00~\copyright~2021 IEEE}
\begin{document}

\maketitle

\begin{abstract}

We study the trade-off between Age of Information (AoI) and maximal leakage (MaxL) in discrete-time status updating systems. A source generates time-stamped update packets that are processed by a server that delivers them to a monitor. An adversary, who eavesdrops on the server-monitor link, wishes to infer the timing of the underlying source update sequence. The server must balance the timeliness of the status information at the monitor against the timing information leaked to the adversary. We consider a model with Bernoulli source updates under two classes of Last-Come-First-Served (LCFS) service policies: (1) Coupled policies that tie the server's deliveries  to the update arrival process in a preemptive queue; (2) Decoupled (dumping) policies in which the server transmits its freshest update according to a schedule that is independent of the update arrivals.
For each class, we characterize the structure of the optimal policy that minimizes AoI for a given MaxL rate. Our analysis reveals that decoupled dumping policies offer a superior age-leakage trade-off to coupled policies. When subject to a MaxL constraint, we prove that the optimal dumping strategy is achieved by dithering between two adjacent deterministic dump periods.

\end{abstract}
\begin{IEEEkeywords}
Age of information, maximal leakage, queueing systems,  status updates, Bernoulli process, Dinkelbach optimization.
\end{IEEEkeywords}

\section{Introduction}

\IEEEPARstart{I}n many applications involving monitoring using sensors such as smart homes/infrastructure, environmental monitoring, and healthcare, a central monitor or observer wishes to track local information at the sensors. 
In the parlance of Age of Information (AoI) analyzes, a \emph{source} (the sensor) generates time-stamped update packets that are processed by a \emph{server} that will forward them to a \emph{monitor} (see Figure \ref{fig:model}).
The utility of these updates is often tied to their timeliness. Stale information at the monitor can lead to incorrect decisions, reduced efficiency, or even critical failures. AoI quantifies the staleness of information: how much time has elapsed since the generation of the most recent update that has been received at the monitor~\cite{kaul2012real}. 

Privacy is another challenge in these sensing systems~\cite{porambage2016quest,lu2018internet}. Even if the data sent is encrypted, the packet timing remains correlated with user activity (the status updates). An adversary who observes packet timings could potentially infer sensitive information about the user~\cite{tahaei2020rise}. For instance, patterns in packet generation or arrival times might reveal when a user is home, their sleep schedule, or the frequency of certain medical events~\cite{apthorpe2017closing,das2016uncovering,buttyan2012traffic}. Approaches to mitigating this privacy loss include inserting transmission delays, altering transmission rates, sending dummy traffic, or otherwise shaping the traffic~\cite{diaz2004taxonomy,chaum1981untraceable,kesdogan1998stop,danezis2004traffic,apthorpe2019keeping,XiongSM:22shaping}. 

In this work we adopt the maximal leakage (MaxL) metric introduced by Issa, Kamath, and Wagner~\cite{IssaKW:16ciss} to quantify the privacy risk. MaxL is a measure of the information leaked about a sensitive quantity $X$ by observing  the output of a function $Y$ that is derived from $X$.  In our system, $X$ represents the packet generation times at the source, which implicitly encode sensitive user patterns. The output $Y$ represents the observed packet delivery timings at the monitor, which server policies modify to protect privacy. The relevance of MaxL and its interplay with AoI can be seen in a smart home setting. For instance, the original generation times of sensor updates, $X$, might reveal when a user is home, their sleep schedule, or appliance usage patterns.  An adversary, observing the obfuscated arrival times $Y$ at the monitor, aims to infer these sensitive activity patterns. 

MaxL is a suitable privacy metric in this context for a few reasons. The MaxL privacy measure has operational significance as it provides a guarantee on the difficulty an adversary faces in guessing a function of the sensitive information. A lower MaxL implies that the observed timings $Y$ provide less advantage to the adversary in  their guessing efforts, thus enhancing privacy~\cite{issa2019operational}. 
We note that MaxL satisfies a data processing inequality~\cite{IssaKW:16ciss}. This property guarantees that any obfuscation to the packet timings by our server cannot increase the information leakage.  

MaxL also offers advantages over other privacy metrics. MaxL is better suited for time-stamped data sequences such as  packet timings, compared to metrics like $k$-Anonymity~\cite{sweeney2002k} which is designed to make individual records indistinguishable from at least  $k-1$ other records.
Similarly, $\epsilon$-differential privacy~\cite{dwork2006calibrating} protects individual contributions in a dataset. In our context, the concern is not about hiding a single packet's existence but about protecting the overall user activity pattern $X$ (original packet timings) from the observation $Y$ (observed packet timings).  

To manage the competing goals of timeliness and privacy, we consider a server that has direct control over when updates are transmitted to the monitor. The fundamental mechanisms at the server's disposal are a controllable \textit{holding time} and an ability to discard updates. An update arriving at the server's buffer is deliberately held for a designed duration before its physical, single-slot transmission to the monitor. This allows for a wide range of strategies. At one extreme, the server could adopt a zero-delay policy, transmitting every update in the slot it becomes available. While this would achieve the lowest possible AoI, it would also create a transmission pattern that maximizes information leakage. 

To mitigate this leakage, the server must strategically manage these holding times to obscure the relationship between the arrivals and the departures. This designed holding period, followed by the physical transmission, constitutes what we will refer to as an update's effective service time. This enables us to analyze the server's privacy-enhancing strategies using queueing theory. The ``service" in our model is not a physical processing requirement but rather an artificial and controllable delay engineered for privacy. The core problem is  thus to design the optimal distribution for these effective service times to balance freshness and privacy.

The main problem we address is the inherent tension between minimizing AoI and reducing the information leakage associated with packet timings. Naively minimizing AoI often leads to revealing transmission patterns, thereby increasing the MaxL. Conversely, aggressively obfuscating packet timings to reduce MaxL can increase the AoI. This paper aims to characterize this fundamental trade-off. While First-Come-First-Served (FCFS) queueing is a common policy, its poor performance for AoI is well-established~\cite{kaul2012status}.  Our primary focus is therefore on policies that prioritize freshness by adopting  Last-Come-First-Served (LCFS) mechanisms that give priority to the freshest updates. We will employ FCFS only as a performance benchmark.
In this context, our paper makes the following contributions:
\begin{itemize}
    \item We formally model the AoI-MaxL trade-off in a discrete-time system in which a source generates updates as a Bernoulli process and the server policy then determines if/when these updates are transmitted to the monitor. 
   
    \item We investigate this trade-off for two classes of LCFS-based service policies, distinguished by whether the server's transmission schedule is {\em coupled} to the arrivals of individual updates:
    \begin{itemize}
        \item LCFS Coupled Policies: 
        The server operates as a Ber/G/1/1 LCFS queue with preemption; a newly arriving update immediately replaces any update currently in service, initiating a new service period. 
        \item LCFS Decoupled (Dumping) Policies: In this class, the server also stores only the single freshest update, but its decision to transmit is decoupled from the arrivals of individual updates. Instead, transmissions are governed by an independent timer. 
    \end{itemize}
    \item For the policy classes, we find optimal strategies that minimize the average AoI at the monitor subject to a MaxL constraint. 
    \item Our analysis provides insights into how system parameters and policy choices influence the balance between information freshness and privacy. 
\end{itemize}

Our analysis reveals several key insights into the age-leakage trade-off. We show that decoupled (dumping) policies outperform coupled policies. For the coupled LCFS policy class, we find that the optimal strategy is a ``greedy" service time distribution that prioritizes the shortest possible service durations. Within the superior decoupled class, we show that while randomized policies are generally suboptimal, the best age-leakage trade-off is achieved by dithering, a policy that probabilistically chooses between deterministic dump periods that differ by one slot in duration. These results show that the most effective approach to balance privacy and AoI combines discarding stale updates with decoupled schedules built upon deterministic principles.  

The remainder of the paper is organized as follows. 
Section \ref{relatedwork} reviews related work in privacy metrics, AoI analysis, and papers  at the intersection of timeliness and privacy. Section \ref{sysmodel} describes the system model. Section \ref{sec:coupled} and Section \ref{sec:decoupled} analyze the AoI-MaxL trade-off for the LCFS coupled and decoupled policies, respectively. In Section~\ref{sec:optimaldecoupled} we derive the optimal policy in the decoupled class. 
Section \ref{results} presents results based on simulations that explore different parameter regimes. In Section \ref{conclusion} we discuss the implications of our results along with different promising directions for future work.

\subsection{Related Work} \label{relatedwork}

The quantification of information leakage is a rich area of study. Classical information theoretic measures such as mutual information (MI)~\cite{shannon1949communication} $I(X;Y)$ capture the average reduction in uncertainty about a secret $X$ given an observation $Y$. Several works  have explored secrecy capacity and utility-privacy trade-offs using MI~\cite{wyner1975wire, gopala2008secrecy, sankar2013utility}. However, MI averages over all possible inputs and outputs and may not protect against scenarios where specific, highly sensitive functions of $X$ can be inferred with high probability~\cite{issa2019operational}. 

An alternate approach employs the maximal leakage (MaxL) metric, proposed by Issa et~al.~\cite{IssaKW:16ciss,issa2019operational}, which provides a worst-case leakage guarantee over all possible functions of $X$ an adversary might try to infer. Several extensions and variations of MaxL have been proposed. Liao et~al.~\cite{liao2020maximal} introduced a more flexible version of MaxL called maximal $\alpha$-leakage where $\alpha\in[1,\infty)$ can be tuned for different applications; for $\alpha=1$ the metric corresponds to MI and as $\alpha\to\infty$, the metric corresponds to MaxL. This framework has been used to study privacy-utility trade-offs for a data disclosure problem under \textit{hard distortion constraints} (which provide deterministic guarantees on fidelity)~\cite{Liaoprivacy2018}. 
More recently, Saeidian et~al.~\cite{saeidian2023pointwise} introduced another variation called Pointwise MaxL. Instead of looking at the average leakage over all possible observed data $Y$, this measure quantifies the information revealed about $X$ by observing a single outcome $Y=y$. 
Other information leakage metrics  studied include maximal correlation~\cite{li2018maximal} and total variation (TV) distance. Rassouli and G\"und\"uz~\cite{RassouliTV2020}   proposed TV distance as a privacy measure, showing that it satisfies key properties like the post-processing inequality, simplifies the privacy-utility optimization into a standard linear program, and provides bounds on other leakage metrics like MI and MaxL.  

Minimizing MaxL subject to various system utility or  cost constraints has also been studied~\cite{saeidian2021optimal, WuOptimalmechanismsLeakage2020, LiaoHypothesisleakage2017, Issaleakage2016}. 
Saeidian et~al.~\cite{saeidian2021optimal} studied the trade-off between MaxL and expected Hamming distortion, showing that the optimal privacy mechanism for a known data distribution involves fully disclosing the most probable symbols while suppressing the least probable ones. 
Liao et~al.~\cite{LiaoHypothesisleakage2017} examined the trade-off between MaxL and the Type-II error exponent as the metric for data utility for hypothesis testing. In the context of the Shannon cipher system, Issa et~al.~\cite{Issaleakage2016} studied MaxL with lossy communication, deriving a single letter characterization for the asymptotic behavior of normalized MaxL. Issa et al.~\cite{issa2019operational} also derived the MaxL for specific queueing models in continuous time, namely an M/M/1 queue and an ``accumulate-and-dump'' system, where packets are held and then released.

The above-mentioned works explore applications of MaxL in various contexts, including cost-based optimization, hypothesis testing under specific utility metrics, and classical communication systems. However, they do not address the problem of information leakage in modern status updating systems where data freshness is a primary concern. For \emph{discrete-time updating systems}, our work addresses this gap by analyzing both queuing and dumping policies with general holding times. We introduce the AoI as a critical performance metric, thereby investigating  age-leakage trade-offs for systems with Bernoulli arrivals.

The time-average AoI  for various systems has been extensively studied since its introduction by Kaul et~al.~\cite{kaul2012real}. Initial work focused on finding the AoI for various continuous time queueing models including the  first-come first-served (FCFS) M/M/1, M/D/1 and D/M/1 queues~\cite{kaul2012real}. With respect to minimizing AoI, the monitor is interested in the most recent updates  and last-come first-served (LCFS) queues, with and without preemption, are typically superior to FCFS queues~\cite{kaul2012status,najm2016age,chen2016age, yates2018agepreemption, bedewy2016optimizing}.  

In the context of rate-constrained age optimization for memoryless sources, \.{I}nan and Telatar~\cite{InanOptimalLabelling2022} studied server policies that optimized the tradeoff between the rate of updates delivered to the monitor and the average age. Their optimal policy waits for a deterministic time $T$ and then transmits the next arriving update. Although this policy utilizes an independent timer, within our framework, we classify it as a coupled policy because the actual transmission instants are tied to the arrival process. In decoupled policies, the transmission instants are determined solely by the schedule and do not wait for fresh arrivals. 

The scope of AoI research has expanded significantly since its introduction more than a decade ago. Various service disciplines and application-motivated system assumptions have been analyzed, including systems with multiple sources~\cite{yates2012real,yates2018age}, energy harvesting sensors~\cite{bacinoglu2017scheduling,farazi2018average, feng2018minimizing,baknina2018coded}, and  different packet management strategies at the source or queue~\cite{CostapacketmanagementAge2014}. tool~\cite{Yates2020TheAO}. 
Another line of research has addressed the problem of scheduling transmissions in multi-user networks to minimize the long-run average AoI~\cite{sun2018age,Hsu2020Scheduling, kadota2018optimizing}. Many of these scheduling problems have been formulated and solved using the theory of Markov Decision Processes~\cite{CeranHARQ, tang2020minimizing, bedewy2021optimal, hsu2018age, sathyavageeswaran2024timely}.  The impact of system parameters, such as the buffer capacity and the number of server on the AoI has also been studied~\cite{soysal2021age, bedewy2019minimizing, kosta2019queue}. In addition to average AoI, the related metric of Peak Age of Information~\cite{CostapacketmanagementAge2014}---the age just before an update is received was introduced. This metric has been analyzed for various scheduling and re-transmission policies in systems with update delivery errors~\cite{chen2016age}, and optimal update dropping policies have been derived to minimize age~\cite{kavitha2021controlling}. For a comprehensive overview, readers are referred to the surveys~\cite{yates2021age, kosta2017age, kahraman2023age}.

The majority of AoI papers focus on continuous time modeling, whereas we adopt a discrete time model in this paper. The work on AoI metrics for discrete time queues~\cite{9148775, tripathi2019age, article, talak2019optimizing} is more directly relevant to our work. In particular our AoI analysis of LCFS policies in Section \ref{sec:coupled} is based on results from Tripathi et al.~\cite{tripathi2019age}. Privacy and AoI for discrete-time systems have been studied under other privacy metrics in recent work, either implicitly or explicitly. Jin et~al.~\cite{jin2020minimizing} studied AoI minimization in mobile crowd-sensing systems, designing a payment mechanism to incentivize mobile agents to share their locations with DP guarantees. Ozel et~al.~\cite{ozel2022state} analyzed the minimization of timing information 
leakage in a discrete-time energy harvesting channel while ensuring status updates remain timely. Zhang et~al.~\cite{zhang2023age} also study a discrete time model and introduced age-dependent differential privacy which takes into account the underlying stochastic evolution of time-varying databases. They show that using stale inputs and/or delaying the outputs can protect privacy. Our work contributes to this emerging field by specifically analyzing the optimal trade-offs between MaxL and AoI.

\section{System Model}\label{sysmodel}

\noindent \textbf{Notation.} 
Random variables are denoted by capital letters, while their specific realizations are denoted by corresponding lower case letters. 
For any binary sequence $u^n = (u_1, u_2, \ldots, u_n)$, we denote its Hamming weight by $\mathrm{wt}(u^n) = \sum_{i = 1}^ n u_i$. We use the notation $[n]$ to denotes the set of integers $\{1,2,\ldots, n\}$. 
\begin{figure}[t]
\centering
\includegraphics[width=\textwidth]{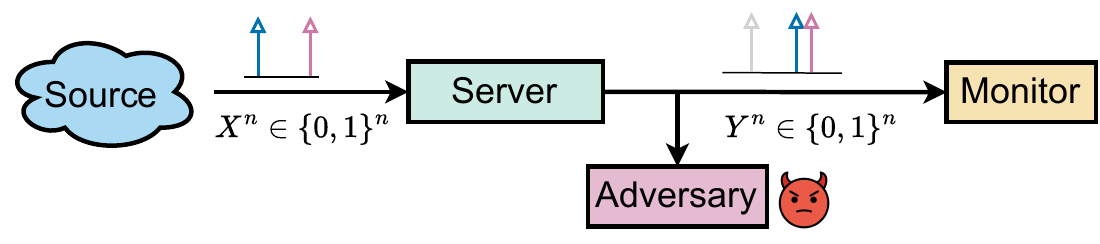}
\caption{Model for our general system setup in discrete time. Updates over $n$ time slots $X^n$ are generated at a source and must be scheduled (causally) to produce an output $Y^n$ observed by an eavesdropper. The mapping $X^n \to Y^n$ must balance the AoI at the monitor and the MaxL to the adversary. \label{fig:model}}
\end{figure}

Our discrete time system consists of a source, server, monitor and an adversary, as shown in Figure \ref{fig:model}. We let $t \ge 0$ denote the time interval $[t,t+1)$. At each time slot, the source may or may not generate an update. Transferring an update from the source to the server takes one time slot. Therefore, an update generated at time $t-1$ is received at the server at the beginning of slot $t$. We let $X_t \in \{0,1\}$ denote the arrivals of updates at the server; $X_t = 1$ indicates the source generated an update at time $t-1$ and otherwise $X_t = 0$. 

Within the time interval $[t,t+1)$ the order of actions is as follows:
    \begin{enumerate}
    \item If $X_t = 1$, the newly arrived update with time-stamp $t-1$ is stored in the buffer. 
    \item The server decides which (if any) of the updates in the buffer will be transmitted according to its policy.
    \item If the server transmits an update, that update  is removed from the buffer and the output for that slot is $Y_t = 1$. Otherwise, if no update is transmitted, $Y_t = 0$.
    \end{enumerate}
The adversary has access to the output $Y^n$ of the system from which information about the input process can be inferred. In this paper we assume the source sends fresh updates as a rate $\lambda$ Bernoulli process. This is both the simplest class of arrival processes and also allows us to simplify the calculation of the MaxL that we define next.

\subsection{Server Policies} \label{sec:policies}

A \emph{server policy} is a sequence of maps $\{ \pi_{t} \colon t \in \mathbb{N} \}$ where $\pi_t \colon \{0,1\}^t \times \{0,1\}^{t-1} \to \{0,1\} \times [t]$ where $[t] = \{1,2,\ldots, t\}$: the server generates $Y^n$ causally from $X^n$ using $(Y_t, t') = \pi_t(X^t, Y^{t-1})$, where $t'$ is the time stamp of an update. If no update is sent, i.e. $Y_t=0$, the time stamp takes a default value $t'=-1$, which is disregarded by the monitor. At time $t$, the adversary observes only $Y_t$ whereas the monitor observes $(Y_t,t')$. We focus on two broad classes of causal policies, depending on the capabilities of the server and the requirements of the monitor.

The goal of the server is to select a policy that balances privacy leakage to the adversary against the freshness of the information at the monitor. We quantify privacy leakage using MaxL~\cite{issa2019operational} and freshness using age of information (AoI)~\cite{kaul2012real}. A server policy determines when and which updates are transmitted from the server's buffer to the monitor. To minimize AoI, it is well known that policies should prioritize the freshest available information, following a LCFS principle. Our analysis therefore focuses on two broad classes of LCFS-based policies, and we compare with the traditional FCFS policy as a benchmark.
 \begin{itemize}
     \item \textbf {LCFS Coupled Policy (Ber/G/1/1)}: In this class, the server's transmission schedule is directly coupled to the arrival of updates. The server operates an LCFS queue with preemption: a newly arriving update immediately discards and replaces any update currently in service. The server is then occupied for a service time $S$ to process and transmit this new update, where $S$ is governed by a designable probability mass function (PMF).
     \item \textbf{ LCFS Decoupled  Policy (Dumping)}: In this class, the server's transmission schedule is decoupled from the update arrival process. The server stores only the single freshest update. Transmissions are governed by an independent timer, and in each designated slot, the server ``dumps" (forwards) the freshest update it currently holds, if any.
 \end{itemize}
For comparison, we also analyze the following baseline policy:
\begin{itemize}
    \item \textbf{FCFS Policy}: In this lossless policy, all arriving updates are placed in a buffer and serviced in the order of their arrival, each with an independent and identically distributed (i.i.d.) service time. 
\end{itemize}
We describe these policies more formally when we analyze them.

\subsection{Information Leakage}

We measure privacy leakage by the MaxL from the source sequence $X^n$ to the server transmission sequence $Y^n$. The updating policy of the source, coupled with the updating policy of the server induces a joint PMF $\pmf{X^n,Y^n}{x^n,y^n}$ on $\{0,1\}^n \times \{0,1\}^n$. We define the support set $\support{X}\triangleq\set{x^n\colon \pmf{X^n}{x^n}>0}$ of the PMF of $X^n$ and the support set $\support{Y}$ of the PMF of $Y^n$ analogously. The definition of MaxL is given as follows.

\begin{definition}[Issa et al.~\cite{issa2019operational}]
Given a joint distribution $P_{X^n Y^n}$ on finite alphabets $\mathcal{X}^n$ and $\mathcal{Y}^n$,  the  MaxL from $X^n\to Y^n$ is
    \begin{align}
    \L(X^n\to Y^n)
    \triangleq\log \sum_{y^n\in \support{Y}}
        \max_{x^n\in\support{X}} 
        \pmf{Y^n|X^n}{y^n|x^n}. 
        \eqnlabel{eqn:maxL_def}
    \end{align}
\end{definition}
In maximal leakage, for each $y^n$ we compute the maximum likelihood estimate of $x^n$ and then take the log of the sum of these likelihoods. 
The key technical challenge in computing \eqnref{eqn:maxL_def} is identifying the the maximum likelihood estimate input sequence $x^n\in \support{X}$ for each possible output sequence $Y^n=y^n$. A key difference between MaxL and privacy metrics like MI that incorporate the distribution of the arrival process is that MaxL depends on the arrival process $X^n$ only through its support $\support{X}$.  While we assume that the arrival process is Bernoulli, our analysis of the leakage applies to all arrival processes satisfying $\support{X}=\set{0,1}^n$ (i.e., with \emph{full support}.) Throughout this paper, logarithms are base 2 so the maximal leakage is measured in bits.

\subsection{Age of Information}

AoI was initially proposed for continuous-time processes; here we use the discrete-time age process model introduced by Kosta et al.~\cite{9148775}. 
In the parlance of AoI, the source arrivals are time-stamped updates  that are sent to the server. We assume the updates arrive at the beginning of a time slot and leave at the end of a time slot. An update generated by the source in slot $t$ (implying $X_{t + 1} = 1$) is based on a measurement of a process of interest that is taken at the beginning of the slot and has time-stamp $t$.  At the end of slot $t$, or equivalently at the start of slot $t+1$, that update will have age $1$. In slot $t+j$, this update will have age $j$. We say one update is fresher than another if its age is smaller. 

An observer of these updates measures timeliness by a discrete-time \emph{age} process $A(t)$ that is constant over a slot and equals the age of the freshest update received prior to slot $t$. When $u(t)$ denotes the time-stamp of the freshest update observed/received prior to slot $t$, the age in slot $t$ is $A(t)=t-u(t)$. We characterize the timeliness performance of the system by the average AoI at the monitor.  

\begin{definition}[Kaul et al.~\cite{kaul2012real}]
A stationary ergodic age process has age of information (AoI) 
\begin{align}
   \E{A(t)}=\limty{T} \frac{1}{T}\sum_{t=0}^{T-1} A(t)\eqnlabel{sum-age},
\end{align} 
where $\E{\cdot}$ is the expectation operator.
\end{definition}

Changing the assumptions on the arrival process does not affect MaxL provided the arrival process has full support. For example this would include sources generated by recurrent Markov chains. However, the AoI will change with the source statistics. An interesting future direction would be to study sources which are run-length limited~\cite{immink1990runlength} as these do not have full support.

\subsection{Trade-offs}

Let $\{X_t : t \in \mathbb{N}\}$ be a binary-valued stochastic process. We say a policy $\bm{\pi} = \{\pi_t : t \in \mathbb{N}\}$ achieves the leakage-age pair $(\Lambda,\Delta)$ if the age process $A(t)$ has AoI $\E{A(t)}=\Delta$ and the induced output sequence $Y^n$ has asymptotic leakage rate 
\begin{align}
    \Lambda = \lim_{n\to \infty} \frac{1}{n}\L(X^n\to Y^n).
\end{align} 

The question we ask in this paper is: what are the $(\Lambda, \Delta)$ pairs achievable by the policies in Section \ref{sec:policies}? 

One classical approach to enhancing privacy for packet-based communications is to insert dummy traffic~\cite{XiongSM:22shaping, fu2003analytical,wang2008dependent}. Under AoI such packets could have a time stamp of $-1$ so that the monitor knows they are fake but an adversary only given timing information could not tell the difference between a fake and a real update/packet. Allowing arbitrary transmissions of dummy packets makes the problem trivial: the output $Y^n$ can be set deterministically to all $1$s by sending a packet in every slot. That is,  in each slot $t$, the newly arrived packet is sent when $X_t = 1$ and a dummy packet is sent if $X_t = 0$.

However, in many practical monitoring systems, such as those using battery-powered sensors, every transmission incurs a significant energy or bandwidth cost. In such resource constrained settings, a policy that transmits in every time slot using dummy packets may be expensive. Hence, we focus here on policies that employ only real updates/packets that can be discarded or delayed by the server but no fake packets can be generated. This frames our problem around the efficient use of limited transmissions to balance the competing goals of freshness and privacy, rather than assuming an unlimited ability to obfuscate traffic at no cost. 

We now turn to analyzing leakage-age trade-offs for a variety of policies. Several results require the condition that the arrival process has full support so that every $x^n \in \{0,1\}^n$ has positive probability. For all of our analyzes, unless otherwise stated, we will assume that the $X_i$  are independent Bernoulli random variable with parameter $\lambda$. That is, packets arrive as a Bernoulli process with rate $\lambda$. 

\section{LCFS Coupled Policies}
\label{sec:coupled}

In this section, we analyze  LCFS Coupled Policies, where the server's transmission schedule is directly coupled to the arrival of updates. This system is modelled as a discrete-time Ber/G/1/1 queue with a preemptive LCFS discipline. A new arrival immediately replaces any update currently in service, ensuring only the freshest update is processed. 

We assume the service times are i.i.d.~random variables denoted by $S$ following a general PMF $g(s)$:
\begin{align}
    P(S=s)=\begin{cases}
    g(s) & s=1,2,3,\cdots \\
    0 & \text{otherwise}.
    \end{cases}
\end{align}
An output sequence $y^n$ is called \textit{achievable} if it can be generated by some input sequence $x^n$. The set of all achievable output sequences, denoted by $\support{Y}$, is constrained by the support of the service time distribution $g(s)$. Specifically, if the minimum possible service time is $s_{1}$, any two departures in an achievable sequence must be separated by at least $s_1$ time slots. That is, if $y_t=1$ and $y_{t'}=1$ for $t < t'$, then $t' - t \ge s_{1}$. Furthermore, assuming the server is idle prior to slot 1, the earliest an update can arrive is slot 0 and the first departure cannot occur before slot $s_{1}$.

\subsection{Shortest Most Probable (SMP) Leakage }

For the subsequent leakage analysis, we focus on the broad and practical class of service PMFs that we call \textit{Shortest Most Probable (SMP) distributions}. A service PMF $g(\cdot)$ is an SMP distribution if its minimum service time, 
\begin{align}
    s_{\min}(g) \triangleq \min \{s: g(s) > 0\},
\end{align} is also its most probable service time satisfying the condition $ g(s_{\min}(g)) \ge g(s)$ for all $s \ge s_{\min}(g)$.
With $\mathcal{P}$ denoting the set of all valid PMFs on the positive integers, the set of all valid SMP distributions, %$\mathcal{G}_\mathrm{SMP}$, 
is defined as
\begin{align}
\mathcal{G}_\mathrm{SMP}\triangleq\set{g\in \mathcal{P}: g(s_{\min}(g))\ge g(s)  \text{ for } s \ge s_{\min}(g)}.
\end{align}
This class is general enough to include common distributions such as a geometric distribution or a uniform distribution over $\{s_{\min}(g),s_{\min}(g)+1,\ldots, s_{\min}(g)+k\}$. We note that the SMP property is preserved when an  SMP distribution is shifted in time. 

The SMP structure is key to simplifying the problem of finding the maximum likelihood input. To find the leakage for policies using SMP service times, we  first identify for each achievable output $y^n$, the input sequence $x^n$ that maximizes the  conditional probability $\pmf{Y^n|X^n}{y^n|x^n}$. The following lemma establishes this.

\begin{lemma}[SMP maximum likelihood input]
\label{maxLlemma}
Consider a server (LCFS w/ preemption or FCFS) with service times drawn from an SMP distribution $g(s)$ with minimum service time $s_1 = s_{\min}(g)$.
For any achievable output sequence $y^n \in \support{Y}$,  
\begin{align}
    \max_{x^n \in \support{X}} \pmf{Y^n|X^n}{y^n|x^n} = [g(s_1)]^{\mathrm{wt}(y^n)}
\end{align}
and this maximum is achieved by a time-shifted input sequence $x^n$ such that for all $i\in [n]$,
\begin{align}
    x_i = \begin{cases}
        1 & \text{if } y_{i + s_1 - 1} = 1,\\
        0 & \text{otherwise.}
    \end{cases} \eqnlabel{time-shifted}
\end{align}
\end{lemma}
Lemma~\ref{maxLlemma} shows that for a given output sequence $y^n$ for the LCFS or FCFS servers, the MaxL is achieved by a time-shifted input  given by \eqnref{time-shifted}. The maximum probability of a single service completion time is $g(s_1)$, a property of the service time distribution itself. Since the queueing discipline (LCFS w/ preemption vs.~FCFS) does not alter this single-slot departure probability for the update in service, the result holds for both policies. The proof follows by an induction argument; details are in Appendix \ref{Lemma: FCFS Leakage}. 

\begin{theorem}[FCFS, LCFS Leakage for SMP Distributions]
    \label{thm:max_leakageFCFSLCFS}
    The MaxL for the LCFS with preemption and FCFS policies with SMP service times $g(s)$  with minimum service time $s_1 = s_{\min}(g)$ is given by 
    \begin{align}
        \L(X^n \to Y^n) = \log  \left( \sum_{k=0}^{\lfloor n/s_1 \rfloor} \binom{n - k(s_1 - 1)}{k} [g(s_1)]^k \right). \eqnlabel{smpleakage}
    \end{align}
\end{theorem}
\begin{proof}
The result follows by applying Lemma \ref{maxLlemma} to the definition of maximal leakage. For either policy, we have 
\begin{align}
    \mathcal{L}(X^n\to Y^n) & = \log \left(\sum_{y^n \in \support{Y}} [g(s_1)]^{\mathrm{wt}(y^n)}\right) \\
    &= \log\left(\sum_{k=0}^{\lfloor n/s_1 \rfloor} N(n,k,s_1)  [g(s_1)]^k \right)\eqnlabel{maxlleakage}, 
\end{align}
where  $N(n,k,s_1) \triangleq \binom{n - k(s_1 - 1)}{k}$ is the number of valid binary sequences of length $n$ with exactly $k$ ones, such that any two ones are separated by at least $s_1-1$ zeros. 
\end{proof}

\begin{proposition}[Bounds on Asymptotic Leakage Rate for the SMP Class]
\label{prop:leakage_rate}
The asymptotic leakage rate $\Lambda$ for policies using an SMP distribution with minimum service time $s_1 = s_{\min}(g)$ is bounded by: 
    \begin{align}
        \frac{1}{s_1} \log (1 +  g(s_1)) \le \Lambda \le \log (1 + g(s_1)). 
    \end{align}
\end{proposition}
\begin{proof}
Let $S_n = \sum_{k=0}^{\lfloor n/s_1 \rfloor} N(n,k,s_1) {[g(s_1)]}^k$. First, we find an upper bound for the sum $S_n$. For any $k$ such that $0\le k\le \lfloor n/s_1\rfloor$ and any $s_1 \ge 1$, $n - k(s_1 - 1) \le n$. Hence
$N(n,k,s_1)\le \binom{n}{k}$, implying 
    \begin{align}
        S_n &\le \sum_{k=0}^{\lfloor n/s_1 \rfloor} \binom{n}{k}[g(s_1)]^k
        %\\&
        \le \sum_{k = 0}^n \binom{n}{k} {[g(s_1)]}^k
        %\\&
        = {(1 + g(s_1))}^ n \eqnlabel{Snbound}.
    \end{align}
 Since $\mathcal{L}(X^n\to Y^n) = \log (S_n)$, it follows from \eqnref{Snbound} that the leakage rate $\Lambda$ satisfies
    \begin{align}
        \Lambda &= \lim_{n\to \infty} \frac{1}{n}\L(X^n\to Y^n) \le \log (1 + g(s_1)) \eqnlabel{upperbound}.
    \end{align}
    Next, we find a lower bound for the sum $S_n$ by counting a fraction of all possible valid sequences in $\support{Y}$.
    We can count only sequences which have a departure at times $s_1, s_1 + 1, \ldots$. There are $\lfloor n / s_1\rfloor$ maximum possible departures in this case. Hence we can write
    \begin{align}
        S_n &\ge \sum _{k = 0}^{\lfloor n/s_1\rfloor}  \binom{\lfloor n/s_1\rfloor}{k} {[g(s_1)]}^k
        = {(1 + g(s_1))}^ {\lfloor n/s_1\rfloor} \eqnlabel{Snlower}.
    \end{align}
     Since $\L(X^n\to Y^n) = \log (S_n)$, it follows from \eqnref{Snlower} that
    \begin{align}
        \L(X^n\to Y^n) &\ge   \lfloor n/s_1 \rfloor \log (1 + g(s_1)).
    \end{align}
    Hence the leakage rate $\Lambda$ satisfies
    \begin{align}
        \Lambda &= \lim_{n\to \infty} \frac{1}{n}\L(X^n\to Y^n) \ge \frac{1}{s_1}\log (1 + g(s_1)) \eqnlabel{lowerbound}.
    \end{align}
    Thus the proposition follows from \eqnref{upperbound} and \eqnref{lowerbound}.
\end{proof}

For $s_1= 1$, the summation in \eqnref{smpleakage} can be solved in closed form  to get the exact leakage:
\begin{align}
    \L(X^n\to Y^n) = \log \left(\sum_{k = 0}^n \binom{n}{k} {[g(1)]}^ k 1^ {n-k}\right) = n \log (1 + g(1)).
\end{align}
In this scenario, the exact leakage rate is $\Lambda = 1 + g(1)$, corresponding to the matching lower and upper bounds for the asymptotic leakage rate.

As a second illustrative example of the leakage calculation, consider the general deterministic policy $g(M)=1$ for some integer $M\ge 2$. This is a valid SMP distribution with $s_1=M$. In this case, any achievable output sequence $y^n$ is constrained such that every $1$ must be followed by at least $M-1$ zeroes. The general formula from Theorem~\ref{thm:max_leakageFCFSLCFS} simplifies to $\L(X^n\to Y^n) = \log (N_n)$, where $N_n$ is the number of such achievable sequences. The count $N_n$ follows the recurrence $N_n = N_{n-1} + N_{n-M}$. For the specific case of $M=2$, this simplifies to the Fibonacci recurrence $N_n = N_{n-1} + N_{n-2}$,  with base cases $N_1=1$ and $N_2=2$. This yields an exact asymptotic leakage rate of $\Lambda = \log (\phi)\approx 0.694$ bits/slot, where $\phi = (1 + \sqrt{5})/2$ is the golden ratio. In this example, both bounds of Proposition~\ref{prop:leakage_rate} are loose. This example also shows that the server can reduce the leakage without a randomized service policy; deterministic service policies that just  reduce the rate of delivered updates also reduce leakage. 

As the Proposition~\ref{prop:leakage_rate} bounds are not tight for general $s_1> 1$,  policies with $s_1>1$ would appear to offer a potential trade-off, decreasing the MaxL rate at the cost of a higher AoI. However, to evaluate this trade-off, we would need to compare the \textit{best policy} with $s_1=1$ against the best policy with $s_1>1$. Our analysis in  Section~\ref{sec:optimallcfs} will show this trade-off is not beneficial. We will prove in Theorem~\ref{thm:optimal_pmfLCFS}, that the LCFS SMP policy that minimizes AoI for any given leakage constraint must have $s_1=1$. This finding will imply that the set of optimal policies for the  $s_1=1$ case strictly dominates any policy with $s_1>1$.  This analytical result is also validated numerically later in Section~\ref{sec:coupledresults}. Therefore, we will focus our analysis on the $s_1=1$ case, for which our leakage rate calculation is exact. 

\subsection{LCFS Age Analysis}
We next derive the average age for the LCFS server with generally distributed service times. From Tripathi et~al.~\cite{tripathi2019age} which derived the average age for the discrete-time LCFS  G/G/1 with preemptive service, the average age is given by\footnote{We note that the formula derived in Tripathi et~al.~\cite[Theorem 3]{tripathi2019age} terminates with $-1/2$ replacing $1/2$ in \eqnref{LCFSTripathi}. In our system setup, an update generated in slot $i - 1$ arrives at the server in the beginning of slot $i$. If it is immediately served and delivered by the end of slot $i$, its age in the beginning of slot $i +1$ is $2$. In the model of Tripathi et~al.~\cite{tripathi2019age}, an update that arrives in slot $i$ can be served and delivered in that same slot, resulting in an age of $1$ at the start of slot $i +1$. This difference in the timing model means the age in our system is consistently one unit greater than the age calculated using the formula in Tripathi et~al.~\cite{tripathi2019age}.}
	\begin{align}
\Delta_{\mathrm{G/G/1}} &= \frac{1}{2}\frac{\E{X^2}}{\E{X}} 
		+ \frac{\E{\min(X,S)}}{P(S\le X)} 
		+ \frac{1}{2} \eqnlabel{LCFSTripathi}, 
	\end{align}
where $X$ and $S$ denotes the independent inter-arrival and service time random variables respectively.
\begin{corollary}[LCFS Ber/G/1 Age]
\label{cor:lcfsage}
    For the discrete-time LCFS Ber/G/1 with preemptive service, the average age is 
    \begin{align}
\Delta_\mathrm{LCFS-Ber/G/1 } &= 1 + \frac{1}{\lambda \E{\lambar ^ {S-1}}},\eqnlabel{LCFSage}
    \end{align}
where $\lambar = 1 - \lambda$.
\end{corollary}
The proof is in appendix \ref{app: corollary1}.
For the specific case of a geometric service time distribution with parameter $\mu$, where $ g(s)= {(1-\mu)}^ {s-1}\mu $, we obtain
\begin{align}
    \Delta_\mathrm{LCFS-Ber/Geo/1 } &= \frac{1}{\lambda} + \frac{1}{\mu}. 
\end{align}
\subsection{Optimal LCFS Coupled Policy}
\label{sec:optimallcfs}

The policies we have analyzed so far, such as those with geometric  service times, are not necessarily optimal. We thus solve the following problem: 
For the class of LCFS coupled polices with SMP service time distributions, what is the best possible trade-off between age and leakage? We want to find the service time PMF $g(s)$ that minimizes the average AoI for a specified MaxL constraint. From Theorem \ref{thm:max_leakageFCFSLCFS}, we see that constraining the leakage is equivalent to fixing the value of $g(s_{\min}(g))$.  Our goal is to find the service time PMF $g(s)$ that minimizes the AoI given by \eqnref{LCFSage} for a specific leakage constraint. Our objective is to find the PMF $g^*_\MLC\in \mathcal{G}_\mathrm{SMP}$ that solves the following optimization problem for a given leakage constraint $\MLC\in (0,1]$:
\begin{align}
    g^*_\MLC = \argmin _{g\in \mathcal{G}_\mathrm{SMP}} \Delta(g) \quad\text{subject to} \quad g(s_{\min}(g)) = \MLC.
\end{align}

The main challenge is that the age expression is a non-linear function of the entire distribution $g(s)$. Intuitively, service should be completed as quickly as possible to minimize AoI. This suggests that the optimal policy should prioritize shorter service durations. The following theorem formalizes this intuition. It establishes that the optimal policy must use the smallest possible minimum service time (i.e. $s_{\min}(g) = 1$) and greedily allocate the maximum possible probability to the shortest service times allowed by the leakage constraint. We also note that the optimal policy dithers between the uniform service policy with support $\{1, \ldots, k\}$ that leaks too much and the uniform service policy with support $\{1, \ldots, k+1\}$ that satisfies the leakage constraint but consequently suffers an age penalty.

\begin{theorem}[Optimal LCFS Policy in the SMP Class]
\label{thm:optimal_pmfLCFS}
Given a leakage constraint $\MLC \in (0, 1]$, the service time PMF $g_\MLC\in \mathcal{G}_\mathrm{SMP}$ that minimizes the average LCFS AoI $\Delta(g)$ has a minimum service time $s_1 = s_{\min}(g) = 1$ and is given by 
\begin{align}
    g_\MLC^*(s) =
\begin{cases}
\MLC & \text{for } s = 1, 2, \ldots, k\\
1 - k\MLC & \text{for } s = k + 1 \\
0 & \text{for } s > k + 1,
\end{cases}
\end{align}
where $k = \floor{1/\MLC}$ is the maximum number of service times that can be assigned probability $\MLC$.
\end{theorem}

We first show that any policy  with a minimum service time greater than 1 can be strictly improved by shifting its entire PMF to start at $s=1$, as this strictly increases the term $\E{\lambar ^{S-1}}$ in Corollary~\ref{cor:lcfsage} and thus decreases age. Then for a policy with $s_1=1$, we show that if it does not have the ``greedy" structure, it can be strictly improved by shifting a small amount of probability mass, $\epsilon$ from a larger service time $s_b$ to a smaller service time $s_a$. This is reinforced by the numerical results presented in Section \ref{sec:coupledresults}.

\section{LCFS Decoupled (Dumping) Policies}
\label{sec:decoupled}
 
We  now analyze  \textit{accumulate-and-dump} server policies. Introduced by Issa et~al.~\cite{issa2019operational},  these decoupled policies operate in two distinct phases. First, the server simply \textit{accumulates} (i.e. stores) updates.  Second, triggered by an independent timer, the server \textit{dumps} stored updates by transmitting them to the monitor. In the original formulation \cite{issa2019operational}, all accumulated updates are dumped. Because including older staler updates in the dump does not reduce the age, we consider here an AoI-based variation in which only the single freshest update received in the first phase is dumped. Hence this variation on accumulate-and-dump can also be viewed as a form of LCFS service since an update arriving at the server preempts any prior received update and then only the last update received is dumped. However, unlike the previously considered preemptive LCFS system, the server's transmission schedule is decoupled from the update arrival process.

We will see this accumulate-and-dump approach is highly effective for minimizing age.  We will first derive the general age and leakage for any random dump schedule and then solve for the optimal schedule that provides the best possible age-leakage trade-off.

\subsection{Random Accumulate and Dump (RAD) Leakage}

In the Random Accumulate and Dump (RAD) policy, the server attempts to produce an output packet at a random number of time steps $D$ after the previous attempt. If one or more updates arrives in those $D$ time steps, the server outputs only the single most recently generated update. If no update arrived in the preceding $D$ time steps, the server outputs nothing. The server's dump attempts are governed by a sequence of i.i.d.~random inter-dump times, $\{D_i\}$, drawn from a general PMF $g(d)$. 
For analytical tractability, particularly for the optimization in Section~\ref{sec:optimaldecoupled}, we assume the inter-dump time distribution $g(d)$ has finite support. This means there is a maximum possible inter-dump time, $d_\mathrm{max}$. Let $D_1$, $D_2$,  $\ldots$, $D_n$ denote i.i.d.~samples of $D$. The policy therefore operates as follows: the server waits $D_1$ time slots and then attempts to dump the most recent packet received, then waits $D_2$ slots and attempts to dump the most recent packet stored and so on. For the RAD policy, the attempted dump time sequence is denoted by the $n$ length binary sequence $U^n=(U_1,U_2,\ldots, U_n)$, such that $U_k=1$ if the server attempts to transmit an update to the monitor in slot $k$. 

To analyze the MaxL  of this policy, we first need to formalize the relationship between a sequence of server outputs $y^n$, and the underlying sequence of dump attempts, $u^n$. We define two sets to capture this relationship from two perspectives.
First, for a fixed output sequence $y^n$, we define the set of all dump attempt sequences $U^n=u^n$ that could have produced it. An output can occur at a certain time only if the server attempts a dump at that time, and this condition is captured by
\begin{align}
    \AU(y^n)&\triangleq\set{u^n\in\{0,1\}^n \mid  u_i=1 \text{ if } y_i=1 \text{ for all } i\in [n]}.
\end{align}
Next, for a fixed sequence of dump attempts $u^n$, we can define the set of all possible outputs. If the server does not attempt a dump at a certain time, no output can occur at that time. This gives us
\begin{align}
    \AY(u^n)&\triangleq\set{y^n\in\{0,1\}^n \mid  y_i=0 \text{ if } u_i=0 \text{ for all } i\in[n]}.
\end{align}

These definitions are logically equivalent: $u^n \in \AU(y^n)\iff y^n\in \AY(u^n)$. This equivalence is a key step for simplifying subsequent calculations and is formally proved in Appendix~\ref{app:RADleakage}. 

To find the MaxL for this service policy we characterize the input sequence $x^n$ that maximizes $\pmf{Y^n|X^n}{y^n|x^n}$ in the following lemma. 

\begin{lemma}[RAD maximum likelihood input]
\label{lemmaRAD}
    Under a RAD policy with a full support arrival process, for each possible output sequence $Y^n=y^n$, a maximum likelihood input sequence is $x^n=y^n$.  
\end{lemma}

\begin{proof}
    Fix an output sequence $Y^n=y^n$. 
    For any arrival sequence $x^n$ we can define 
    a deterministic function $Q(x^n,u^n)$ such that $y^n = Q(x^n,u^n)$. 
    Here, $Q(x^n,u^n)$ is the function that deterministically computes the output $y^n$ by applying the RAD policy to the inputs $x^n$ and $u^n$. Specifically, $y_i=1$ if and only if a dump is attempted ($u_i=1$) and at least one new update (an $x_j=1$) has arrived in the interval since the last dump attempt (or since $i=0$). 
    We can then define the set of attempted dump sequences that are consistent with $x^n$ and $y^n$ as 
    \begin{align}
        B(x^n, y^n)\triangleq\set{u^n \mid Q(x^n,u^n)=y^n}.
    \end{align}
    We have
    \begin{align}
    B(x^n,y^n)\subseteq \AU(y^n).
    \eqnlabel{Bxysubset}
    \end{align}
    If $x^n = y^n$, every $u^n\in \AU(y^n)$ is consistent with the input-output pair $x^n,y^n$ so 
    \begin{align}
B(y^n,y^n)=\AU(y^n).\eqnlabel{Byy}
\end{align}
Let $P_D(u^n)$ denote the probability of the specific sequence of dump attempts $u^n$, which is determined by the inter-dump distribution $g(d)$. 
If $x^n\neq y^n$, then there exists a $u^n \in \AU(y^n)$ such that $u^n \notin B(x^n,y^n)$. In this case we have an upper bound on the conditional probability:
  \begin{align}
        P(y^n|x^n)&=\sum_{u^n\in B(x^n,y^n)} P_D (u^n)
        %\\&
        \le \sum_{u^n\in \AU(y^n)}P_D (u^n).
\eqnlabel{PDupper}
    \end{align}
It follows from \eqnref{Bxysubset} and \eqnref{Byy} that $x^n=y^n$ achieves the upper bound in \eqnref{PDupper}.   Thus $x^n=y^n$ is the input sequence that maximizes $\pmf{Y^n|X^n}{y^n|x^n}$ over all possible arrival sequences. 
\end{proof}

\begin{theorem}[RAD Leakage]
    \label{thm: RADleakage}
    For the RAD policy, the MaxL is given by 
    \begin{align}
        \L(X^n\to Y^n)&= \log \left(\E{2^{\mathrm{wt}(u^n)}}\right), \eqnlabel{RADleakage}
    \end{align}
    where $\mathrm{wt}(u^n)$ is the number of attempted dumps in $n$ time slots.  
\end{theorem}
This result is obtained by re-writing the sum in \eqnref{eqn:maxL_def}, swapping the order of summation and then grouping terms according to their Hamming weight. The proof is provided in Appendix~\ref{app:RADleakage}.

\subsection{Random Accumulate and Dump (RAD) Asymptotic Leakage Rate}
For large $n$, we now derive the asymptotic leakage rate of the RAD server using tools from renewal theory~\cite{lalleyRenewal}. Let $m(n) = \E{2^{K_n}}$ , where $K_n=\sum_{i=1}^n U_i$ is the number of attempted dumps in $n$ slots. Let the first dump attempt occur at $D_1=d$ . From the law of total expectation we have
\begin{align}
    m(n) = \E{2^{K_n}} &= \E{\E{2^{K_n}|D_1}}. \eqnlabel{totalexp}
\end{align}

If $d> n$, no dump attempt happened in $n$ slots and we have $K_n = 0$, which implies 
\begin{align}
    \E{2^{K_n}} = 1. \eqnlabel{exp1}
\end{align}
If $d \le n$, the first dump happens at $d$, and the total number of dumps is given by one plus the remaining dumps in $n - d$ slots. For $d\le n$, we can write
\begin{align}
    \E{2^{K_n}| D_1 = d} & = \E{2 ^{1 + K_{n - d}}}
    = 2 m(n - d). \eqnlabel{firstdump}
\end{align}

Substituting \eqnref{exp1} and \eqnref{firstdump} into \eqnref{totalexp} yields
\begin{align}
    m(n) &= \sum_{d=1} ^ n g(d) \cdot 2 m(n - d) + \sum_{d = n + 1}^ \infty g(d) \cdot 1\\
    &= 2 \sum_{d = 1}^ n g(d) m(n - d) + P(D> n) .\eqnlabel{recursion}
\end{align}
 
To solve for $m(n)$, we use generating functions defined as follows:
\begin{align}
    M(z) &\triangleq \sum_{n = 0} ^\infty m(n) z^n,\\
    G(z)& \triangleq \sum_{d = 1}^{d_\mathrm{max}} g(d) z^ d,\\
    F(z) &\triangleq \sum_{n = 0}^ \infty P(D> n) z^n.
\end{align}
With these generating functions, \eqnref{recursion} becomes
\begin{align}
    \sum_{n = 0 }^\infty m(n) z^ n &= 2\sum _{n = 0 }^\infty \left(\sum_{d = 1}^n g(d) z^ d \cdot m(n - d) z^{n - d}\right) + \sum_{n = 0}^ \infty P(D > n) z^ n .
\end{align}
This implies
\begin{align}
    M(z) &= 2 \sum_{d = 1}^ \infty g(d) z^d \sum _{n = d}^ \infty m(n - d ) z^{n - d} + F (z)\\
    &= 2 \sum _{d = 1}^ \infty g(d) z^ d \sum _{j = 0}^ \infty m(j) z^ j + F (z)\\
    &= 2 G(z) M(z) + F (z).
\end{align}
Now we solve for $M(z)$:
\begin{align}
    M(z) &= \frac{F(z)}{1 - 2 G(z)} \eqnlabel{tf}. 
\end{align}

Let $z_0$, $z_1, \ldots, z_{d_\mathrm{max} - 1}$ be the roots of $\E{z^{-D}}= 1/2$. The asymptotic behavior of $m(n)$ is determined by these roots $z_i$, as they correspond to the poles $p_i=1/z_i$ of $M(z)$ (where $G(p_i) = 1/2$).  

\begin{lemma}[Properties of the Roots of $\E{z^{-D}} = 1/2$]
\label{lemma:unique}
    Let $f(z) = \E{z^{-D}}= \sum_{d = 1}^ \infty z^{-d} g(d)$. The equation $f(z) = 1/2$ has a unique positive real root $z_0$, and this root has the maximum magnitude among all roots (real or complex). 
\end{lemma}
The proof is in Appendix~\ref{app:unique}. With Lemma~\ref{lemma:unique}, the generating function $M(z)$ can be used to analyze the underlying renewal process associated with dumping, yielding the next result.

\begin{theorem}
[Asymptotic Leakage Rate of RAD]
\label{thm:RADleakagerate}
    For the RAD policy, with a generally distributed inter-dump time $D$, the asymptotic MaxL rate $\Lambda$ is given by 
    \begin{align}
        \Lambda = \log (z_0)
    \end{align}
where $z_0$ is the unique positive real root of the equation $\E{z^{-D}} = 1/2$.
\end{theorem}

The proof is provided in Appendix~\ref{app: RADleakagerate}. This theorem reveals the explicit relationship between the timing of a RAD policy, as specified by the PMF of $D$,  and the resulting privacy loss (the rate $\Lambda$). The relationship $\Lambda = \log(z_0)$ means minimizing leakage is equivalent to minimizing $z_0$. The equation  $\E{z^{-D}} = 1/2$ defines $z_0$ via a strictly decreasing function $f(z) = \E{z^{-D}}$. We achieve a smaller $z_0$ and thus less leakage by making $D$ stochastically larger, i.e., dumping less frequently. Thus we have the fundamental trade-off that stochastically larger $D$ reduces leakage but increases the age, as we will show later in Theorem~\ref{thm: RAD age}. 

\subsection{Random Accumulate and Dump (RAD) Age}

We employ the sampling of age processes approach~\cite{Yates2020TheAO} for the age analysis of the RAD policy. To use this framework, we use a different model than we did for the leakage analysis.  For the purpose of AoI analysis, if no update arrives before the dump attempt, then the server resends the previously dumped update, which has been called a ``fake update'' \cite{yates2018agepreemption}. With respect to the age at the monitor, in the absence of an update to dump, it makes no difference if the server sits idle  or if the server repeats sending the prior update. Neither changes the age at the monitor. It is important to note that this model is used exclusively for the age analysis. The leakage calculation is strictly based on the physical model where an empty dump attempt results in no transmission ($Y_t = 0$), as this is the channel observed by the adversary.  However this fake update approach %greatly 
simplifies the age analysis.  As observed in \cite{Yates2020TheAO}, we can now view the monitor as \emph{sampling} the update process of the server. With each dump instance, the monitor receives the freshest update of the server,  resetting the age at monitor to the age of the dumped update. In the same way, we can view the server as  sampling the always-fresh update process of the source. 

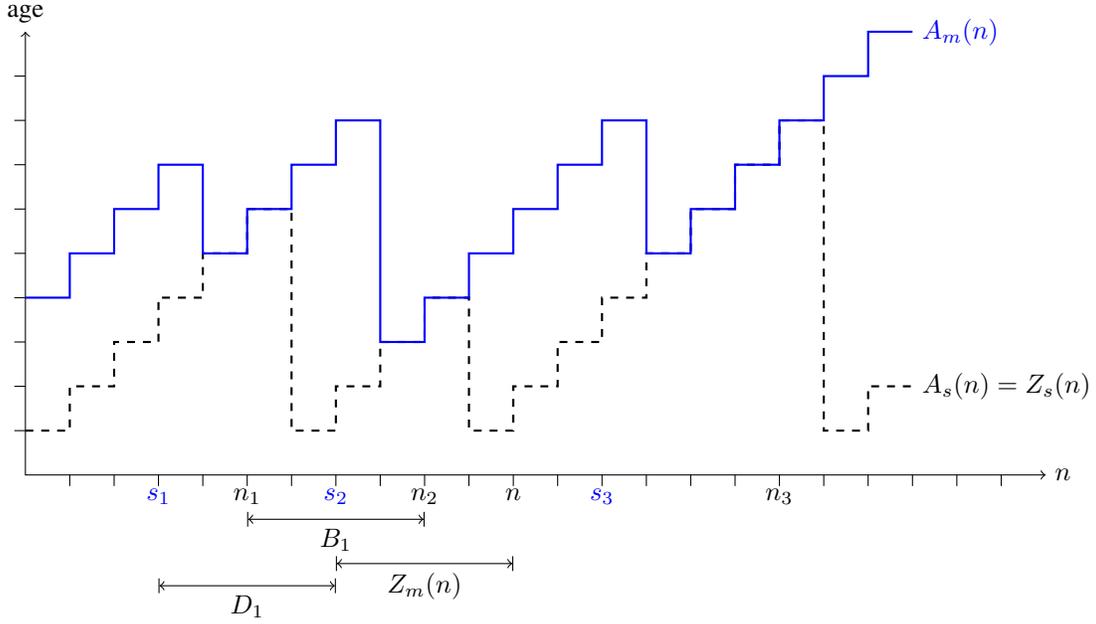
\begin{figure}[t]
    \centering
\begin{tikzpicture}[scale=\linewidth/28cm]%Sampler 
\draw [<->] (0,10) node [above] {age} -- (0,0) -- (23,0) node [right] {$n$};
\draw [thin]\xtic{1}\xtic{2}\xtic{3}\xtic{4}\xtic{5}\xtic{6}\xtic{7}
\xtic{8}\xtic{9}\xtic{10}\xtic{11}\xtic{12}\xtic{13}\xtic{14}\xtic{15}\xtic{16}\xtic{17}\xtic{18}\xtic{19}\xtic{20}\xtic{21}\xtic{22}; 
\draw
[thin]\ytic{1}\ytic{2}\ytic{3}\ytic{4}\ytic{5}\ytic{6}\ytic{7}\ytic{8}\ytic{9};
\draw [thick, dashed] %input age
(0,1) -- ++(1,0) -- 
++(0,1) -- ++(1,0) --
++(0,1) -- ++(1,0) --
++(0,1) -- ++(1,0) --
++(0,1) -- ++(1,0) --
++(0,1) -- ++(1,0) --
++(0,-5)-- ++(1,0) -- % update #1
++(0,1) -- ++(1,0) --
++(0,1) -- ++(1,0) --
++(0,1) -- ++(1,0) --
++(0,-3)-- ++(1,0) -- %update #2
++(0,1) -- ++(1,0) --
++(0,1) -- ++(1,0) --
++(0,1) -- ++(1,0) --
++(0,1) -- ++(1,0) --
++(0,1) -- ++(1,0) --
++(0,1) -- ++(1,0) --
++(0,1) -- ++(1,0) --
++(0,-7) -- ++(1,0) -- %update #3
++(0,1) -- ++(1,0) node [right] {$A_s(n) = Z_s(n)$};
\draw [thick,blue] %monitor age
(0,4) -- ++(1,0) -- 
++(0,1) -- ++(1,0) --
++(0,1) -- ++(1,0) --
++(0,1) -- ++(1,0) --
++(0,-2)-- ++(1,0) -- % sample 1
++(0,1) -- ++(1,0) --
++(0,1) -- ++(1,0) --
++(0,1) -- ++(1,0) --
++(0,-5)-- ++(1,0) -- % sample 2
++(0,1) -- ++(1,0) --
++(0,1) -- ++(1,0) --
++(0,1) -- ++(1,0) --
++(0,1) -- ++(1,0) --
++(0,1) -- ++(1,0) --
++(0,-3)-- ++(1,0) -- %sample 3
++(0,1) -- ++(1,0) --
++(0,1) -- ++(1,0) --
++(0,1) -- ++(1,0) --
++(0,1) -- ++(1,0) --
++(0,1) -- ++(1,0) node [right] {$A_m(n)$};
\draw 
(5,-0.1) node [below] {$n_1$} 
(9, -0.1) node [below] {$n_2$} 
(11,-0.1) node [below] {$n$} 
(17,-0.1) node [below] {$n_3$};
\draw 
(3,-0.1) node [below,blue] {$s_1$} 
(7,-0.1) node [below,blue] {$s_2$} 
(13,-0.1) node [below,blue] {$s_3$};
%(15,0) node [below,blue] {$s_4$};
\draw [|<->|] (7,-2) to node [below] {$Z_m(n)$} ++(4,0);

\draw [|<->|] (5,-1) to node [below] {$B_1$} (9,-1); % Y1 = n2-n1 = 9-5=4
\draw [|<->|] (3,-2.5) to node [below] {$D_1$} (7,-2.5); % Y1' = s2-s1 = 7-3=2
\end{tikzpicture}
\caption{The source sends fresh updates to the server in slots $N_k=n_k$, inducing the age process $A_s(n)$ at the input to the server. The server sends samples of the most recent update to the monitor in time slots $S_k=s_k$, inducing the age process $A_m(n)$ at the monitor.}
\label{fig:DTsampler}
\end{figure}
\subsubsection{Sampling the Source: The Age of Fresh Updates}
\label{sec:fresh}
Referring to Figure~\ref{fig:model}, the source can generate a fresh (age zero) update in a slot $n$ and forward it to the server in that same slot. This fresh update arrives at the server at the end of slot $n$ with age $1$.  The age process $A_s(n)$ of an observer at the server input is reset at the start of slot $n+1$ to $A_s(n+1)=1$. This process $A_s(n)$ is the age of the source's renewal process, which is labeled as $Z_s(n)$ in Figure~\ref{fig:DTsampler}.
When a source generates fresh updates in slots $N_k=n_k$, the age $Z_s(n)$ evolves as the sequence of staircases shown in Figure~\ref{fig:DTsampler}. The source update times $N_k$ form a renewal process with i.i.d.~ interarrival times $B_k = N_k - N_{k-1}$. 
Defining the indicator $\I{A}$ to be $1$ if event $A$ occurs and zero otherwise, we can employ Palm probabilities~\cite{BacelliBremaudBook} to calculate the age PMF 
\begin{align}
P(A_s(n)=a)=\limty{N}\frac{1}{N}\sum_{n=0}^{N-1} \I{A_s(n)=a}.\eqnlabel{PalmZ}
\end{align}
The sum on the right side of \eqnref{PalmZ} can be accumulated as rewards over each renewal period. In the $k$th renewal period, we set the  reward $R_k$ equal to the number of slots in the $k$th renewal period $k$ in which $A_s(n)=a$. When the renewal period has duration $B_k\ge a$, there exists exactly one slot $n'$ in the renewal period in which $A_s(n')=a$ and thus the reward is $R_k=1$; otherwise $R_k=0$.  Thus, for $a=1,2,\ldots$,  $R_k=\I{B_k\ge a}$. From renewal reward theory, it follows from \eqnref{PalmZ} that
\begin{align}\eqnlabel{Zpmf}
    P(A_s(n)=a)=\frac{\E{R_k}}{\E{B_k}}=\frac{P(B_k\ge a)}{\E{B_k}},\quad a=1,2,\ldots.
\end{align}
This is the discrete-time version of the well-known distribution of the age of a renewal process. It follows from \eqnref{Zpmf} that the average age of the current update at the server has average age
\begin{align}
  \E{A_s(n)}&=\sum_{a=1}^\infty a P(A_s(n)= a)=\frac{\E{B_k^2}}{2\E{B_k}}+\frac{1}{2}.
  \eqnlabel{EZ-input-age}
\end{align}

\subsubsection{Sampling the Server: The Age at the Monitor}
Now we examine age at the monitor for the RAD server. As depicted in Figure~\ref{fig:DTsampler}, the server (%whether RAD, DAD or RAD)
sends its freshest update to the monitor at sample times $S_1,S_2,\ldots$.  These sample times $S_k$ also form a renewal process with i.i.d.~inter-sample times $D_k=S_k-S_{k-1}$. We now analyze the average age at the monitor in terms of the moments of $D_k$. 

The age process $A_s(n)$ at the input to the server is as defined in Section~\ref{sec:fresh}. When the server sends its most recent update to the monitor in slot $s_k$, this update has age $A_s(s_k)$ at the start of the slot. The monitor receives the update at the end of the slot with age $A_s(s_k)+1$. Thus, at the start of slot $s_k+1$, the age at the monitor is reset to $A_m(s_k+1)=A_s(s_k)+1$.  Graphically, this is depicted in Figure~\ref{fig:DTsampler}. 

To describe the monitor age $A_m(n)$ in an arbitrary slot $n$, we look backwards in time and define $Z_m(n)$ as the age of the renewal process defined by the inter-renewal times $D_k$ associated with sampling the server.   In Fig.~\ref{fig:DTsampler} for example,  in slot  $n=11$, the last update was sampled by the server at time $s_2=7$ and $Z_m(11)=11-7=4$. We note that the $Z_m(n)$ process is the same as the $A_s(n)$ process, modulo the inter-renewal times now being labeled $D_k$ rather than $B_k$. In particular the PMF $P(Z_m(n) = a)$ and expected value $\E{Z_m(n)}$ are described by \eqnref{Zpmf} and \eqnref{EZ-input-age} with $D_k$ replacing $B_k$. Hence we obtain
\begin{align}
    \E{Z_m(n)} &= \sum_{a=1}^\infty a P(Z_m(n) = a) = \frac{\E{{D_k}^2}}{2\E{D_k}} + \frac{1}{2}\eqnlabel{expZmn}.
\end{align}

The age at the monitor in slot $n$  is then 
\begin{align}
    A_m(n)=A_s(n- Z_m(n))+Z_m(n).
\end{align}
When the age process $A_s(n)$ at the input to the server is stationary and the sampling process  that induces $Z_m(n)$ is independent of the $A_s(n)$ age process, it follows that $A_s(n)$ and $A_s(n-Z_m(n))$ are identically distributed. Thus,
the average age at the monitor is
\begin{align}
    \E{A_m(n)}&=\E{A_s(n-Z_m(n))}+\E{Z_m(n)}\nn&=\E{A_s(n)}+\E{Z_m(n)} \eqnlabel{samplingeqn}.
\end{align}
The inter-arrival times $B_k$ and the inter-sample times $D_k$ are sequences of i.i.d.~random variables, we can therefore drop the index $k$ when considering their moments. Employing  \eqnref{EZ-input-age} and \eqnref{expZmn} to evaluate both $\E{A_s(n)}$ and $\E{Z_m(n)}$, we obtain
\begin{align}
    \E{A_m(n)}= \frac{\E{B^2}}{2\E{B}}+ \frac{1}{2}+ \frac{\E{D^2}}{2\E{D}} +\frac{1}{2} \eqnlabel{agesampling}.
\end{align}
The source generates packets as a rate $\lambda$ Bernoulli process with 
$\E{B}=1/\lambda$ and $\E{B^2}=(2-\lambda)/\lambda^2$. With these observations, \eqnref{agesampling} yields the following theorem. 
\begin{theorem}[RAD Age]
\label{thm: RAD age}
When the source emits updates as a rate $\lambda$ Bernoulli process, the average age at the monitor for the RAD policy is 
\begin{align}
    \E{A_m(n)}&= \frac{1}{\lambda}+ \frac{\E{D^2}}{2 \E{D}} + \frac{1}{2}.
\end{align} \eqnlabel{Radagetheorem}
\end{theorem}

\subsection{Analysis of Special RAD Policies}
Here we evaluate three simple RAD policies, corresponding to deterministic, geometric, and discrete uniform inter-dump times. For each policy, parameters are chosen so that $\E{D}=\tau$ and thus the average dump attempt rate is always $1/\tau$. 

The age analysis of these specific policies is a direct application of Theorem~\ref{thm: RAD age}. We can calculate the AoI for each distribution by substituting its specific first and second moments $\E{D}$ and $\E{D^2}$ into the general age formula~\eqnref{Radagetheorem}. The results are summarized in the following corollary.

\begin{corollary}[AoI of Special RAD Policies]
\label{Cor: Special-RAD-age}
When the source emits updates as a rate $\lambda$ Bernoulli process, the average age at the monitor for the three special RAD policies, each parameterized by its mean inter-dump time $\E{D}=\tau$, are
\begin{align}
    \E{A_m(n)}&=\begin{cases}
    1/\lambda + (\tau + 1)/2 & \text{DAD server,}\\
    1/\lambda + (2\tau + 1)/3 & \text{Uniform RAD server,}\\ 
    1/\lambda + \tau & \text{Geometric RAD server}.
    \end{cases}
\end{align}
\end{corollary}
For any given mean inter-dump time $\tau > 1$, the zero-variance DAD policy achieves the lowest AoI. The Geometric RAD policy, which has the highest variance ($\mathrm{Var}(D) = \tau(\tau-1)$), performs the worst. The Uniform RAD policy, with a moderate variance ($\mathrm{Var}(D) = (\tau^2-1)/3$), has an AoI that lies between these two extremes. However, this age perspective is not the complete story. From a leakage perspective, the policies may employ different values of $\tau$ to meet a MaxL constraint.  We now derive the leakage for each of these RAD policies.  

\subsubsection{Deterministic Accumulate-and-Dump (DAD)} 
For the Deterministic Accumulate-and-Dump (DAD), the server dumps the freshest update after every $\tau$ slots.  
In this policy, the output sequence $Y^n$ is  a deterministic function of the input $X^n$. Defining $K=\floor{n/\tau}$, $Y^n$ has the form 
\begin{align}\eqnlabel{Yn-DADsupport}
Y^n=(0^{\tau-1},Y_\tau,0^{\tau-1},Y_{2\tau},\ldots,Y_{K\tau},0^{n-K\tau}),
\end{align}
where $Y_{k\tau}=0$ if and only if $(X_{(k-1)\tau+1},\cdots, \X_{k\tau})=0^\tau$  and otherwise  $Y_{k\tau}=1$. This structure simplifies the leakage calculation.

\begin{theorem}[DAD Leakage]
\label{thm: DAD-leakage}
The DAD policy has an asymptotic MaxL rate given by 
\begin{align}
\Lambda 
&=\frac{1}{n}\L(X^n\to Y^n) = \frac{1}{\tau}. 
\eqnlabel{DADLeakage}
\end{align}
\end{theorem}

\begin{proof}

For the DAD policy, the output $y^n$ is a deterministic function of the input $x^n$. For any achievable output sequence $y^n \in \support{Y}$, the input $x^n = y^n$ (where an update arrives only in the slot of a dump) yields $\pmf{Y^n|X^n}{y^n|y^n}=1$. Thus we have
\begin{align}
    \max_{x^n\in\support{X}} \pmf{Y^n|X^n}{y^n|x^n} = 1, \quad \text{for all } y^n \in \support{Y}.
\end{align}

The MaxL calculation \eqnref{eqn:maxL_def} thus simplifies to counting the number of achievable output sequences as follows
\begin{align}
    \L(X^n\to Y^n) &= \log \sum_{y^n\in \support{Y}} 1 = \log |\support{Y}|.
\end{align}
Since $Y^n$ has the form \eqnref{Yn-DADsupport}, the support set $\support{Y}$ has size $\abs{\support{Y}}=2^K=2^{\floor{n/\tau}}$. This yields 
\begin{align}
    \L(X^n\to Y^n)&=\log {2}^{\left \lfloor{n/\tau}\right \rfloor} \eqnlabel{dump_instance}
    %\\&
    =\floor{\frac{n}{\tau}}.
    %\nonumber %
\end{align}
Since $\limty{n}\floor{n/\tau}/n=1/\tau$, the claim \eqnref{DADLeakage} follows.
\end{proof}

\subsubsection{RAD with Geometric Dumps}

For this policy, the time $D$ between dump attempts follows a geometric distribution. To achieve a mean of $\E{D} = \tau$, the underlying Bernoulli dump attempt probability must be $\mu = 1/\tau$.

\begin{corollary}[Geometric RAD Leakage]
\label{cor:RADGeo}
    If the server employs a RAD policy where the time $D$ between dump attempts is geometrically distributed with  expected value $\E{D}=\tau$, the asymptotic MaxL rate is given by 
    \begin{align}
        \Lambda 
        %\frac{1}{n}\L(X^n\to Y^n) 
        &= \log (1 + 1/\tau).
    \end{align}

\end{corollary}
The proof is in Appendix~\ref{app:RADgeo}.

The DAD policy's leakage rate is $\Lambda_{\mathrm{DAD}} = 1/\tau$. For the Geometric RAD policy, the rate is $\Lambda_{\mathrm{Geo}} = \log(1 + 1/\tau)$. We can compare these two for all $\tau \ge 1$:
\begin{itemize}
    \item At $\tau=1$, both policies are identical (a dump attempt in every slot), and their leakage rates are equal.
    \item For all $\tau > 1$, the variable $x = 1/\tau$ is in the interval $(0, 1)$. We have $\log(1+x) > x$ for all $x \in (0,1)$.
\end{itemize}
This implies $\Lambda_{\mathrm{Geo}} > \Lambda_{\mathrm{DAD}}$ for all $\tau > 1$. Thus, for any mean inter-dump time $\tau > 1$, the Geometric RAD policy always has a strictly higher (worse) leakage rate than the deterministic DAD policy.

\subsubsection{RAD with Uniform Dumps}
For this policy, the inter-dump time $D$ is drawn from a discrete uniform distribution over the set $\{1, 2, \ldots, k\}$ for some integer $k \ge 1$. To achieve a mean inter-dump time of $\E{D} = (k+1)/2 = \tau$, we must have $k = 2\tau - 1$. This implies that this policy is only defined for $\tau$ values that are integers or half-integers (i.e., $\tau \in \{1, 1.5, 2, 2.5, \ldots\}$).

\begin{corollary}[Uniform RAD Leakage]
\label{cor:RADunifleakage}
    If the server employs a RAD policy where the time between dump attempts is uniformly distributed with mean $\E{D}=\tau$, the asymptotic MaxL rate is given by $\Lambda = \log(z_0)$, where $z_0$ is the unique positive real root of the equation
    \begin{align}
        \frac{1- z_0^{-(2\tau-1)}}{(2\tau-1)(z_0 - 1)} = \frac{1}{2}.
    \end{align}
\end{corollary}

The proof is in Appendix~\ref{app:RADunifleakage}.
 
\section{Optimal Decoupled Policy}
\label{sec:optimaldecoupled}

Here  we find the optimal dump policy, which is characterized by the pmf $g(d)$ of the attempted inter-dump time $D$, that minimizes the age for a given fixed asymptotic MaxL rate. The leakage rate constraint is given by $\L _\mathrm{rate}(g) = \Lambda$. From Theorem~\ref{thm:RADleakagerate}, a RAD policy with PMF $g(d)$ achieves a target leakage rate of $\Lambda$ when $z_0= 2 ^ \Lambda$ satisfies the condition
\begin{equation}
    \E{z_0^{-D}} = \sum_{d=1}^{\infty} g(d) z_0^{-d} = \frac{1}{2}.
\end{equation}
 Since Theorem~\ref{thm: RAD age} implies minimizing $\Delta(g)$ is equivalent to minimizing the fractional term $\E{D^2}/\E{D}$, the optimization problem is
\begin{subequations} \eqnlabel{eq:OptimalDumpProblem}
\begin{align}
    \gamma^* = \min_{g(d)} \quad &  \frac{\E{D^2}}{\E{D}}  \eqnlabel{eq:Objective} \\
    \text{s.t.} \quad & \sum_{d=1}^{\infty} g(d) z_0^{-d} = \frac{1}{2} \eqnlabel{eq:LeakageConstraint} \\
    & \sum_{d=1}^{\infty} g(d) = 1, \quad g(d) \ge 0 .\eqnlabel{eq:ProbabilityConstraint}
\end{align}
\end{subequations}

We will see that the solution to \eqnref{eq:OptimalDumpProblem} is an optimal decoupled policy $g^*(d)$ that we call  \textit{Dithering DAD} (D-DAD) because it dithers between the DAD policy with inter-dump time $\floor{1/\Lambda}$ that leaks too much and the DAD policy with inter-dump time $\ceiling{1/\Lambda}$ that satisfies the leakage constraint but consequently suffers an age penalty. 

\begin{theorem}[Dithering DAD (D-DAD) Policy]
\label{thm:optimaldadts}
    Given a target leakage rate $\Lambda > 0$, the decoupled policy $g^*$ that minimizes the average AoI $\Delta_{\mathrm{RAD}}(g)$  is supported only at $i = \floor{1/\Lambda}$ and $j=i + 1$. The policy is zero everywhere else. With $z_0 = 2^\Lambda$, the non-zero probabilities $p_i = g^*(i)$ and $p_j = g^*(j)=1-p_i$ are found by solving $p_i z_0^{-i} + p_j z_0^{-j} = 1/2$.     
\end{theorem}
We note that if $1/\Lambda$ is an integer, then $i = 1/\Lambda$, which implies $z_0^{-i} = 1/2$, and D-DAD  policy reduces to the DAD policy with $p_i=1$ and $p_j=0$. The remainder of this section is dedicated to proving this optimality.

Since the optimization problem in \eqnref{eq:OptimalDumpProblem} has a non-convex fractional objective, we use the Dinkelbach fractional programming method~\cite{dinkelbach1967nonlinear} to transform it into an equivalent average-cost problem. We note that the method requires the inter-dump time to be bounded by a large finite value $d_\mathrm{max}$. We will see that $d_{\max}$ has no impact on the optimization because a policy that puts non-zero probability on a very large inter-dump time is AoI suboptimal.
With the $d_{\max}$ limitation, the Dinkelbach transformation is valid because our problem satisfies a set of key conditions that we verify in Appendix~\ref{app: dinkelbach}. 

We start by defining $J(\gamma) \triangleq \min_g \E{D^2} - \gamma \E{D}$, where $\gamma$ is the Dinkelbach variable. 

The following lemma formally guarantees the equivalence between solving the original problem in \eqnref{eq:OptimalDumpProblem} and finding the specific parameter $\gamma^*$ for which  $J(\gamma^*)= 0$. 

\begin{lemma}[\cite{dinkelbach1967nonlinear}, Theorem]
    The problem of finding $\min_g \{\E{D^2}/\E{D}\}$ is equivalent to finding the value $\gamma^*$ that solves the problem
    \begin{align}
        J(\gamma^*) = \min_{g} \left\{ \E{D^2} - \gamma^* \E{D} \right\} = 0,
    \end{align}
    where the minimization is performed over all PMFs $g(d)$ that satisfy the constraints \eqnref{eq:LeakageConstraint} and \eqnref{eq:ProbabilityConstraint}.
\end{lemma}

While Dinkelbach's method typically requires an iterative line search to find the value of $\gamma^*$ that drives $J(\gamma^*)$ to zero, the specific structure of our problem allows for a more direct, analytical solution. 

For any given $\gamma$, we define the cost for each possible dump duration $d$ as
\begin{align}
    C_\gamma(d) \triangleq d^2 - \gamma d.
\end{align}
For any given $\gamma$, we first solve the following problem:
\begin{subequations}
    \begin{align}
    \min_{g(d)} \quad &  \sum_{d=1}^{d_{\max}} g(d) C(d) %\eqnlabel{obj}
    \\
    \text{s.t.} \quad & \sum_{d=1}^{d_{\max}} g(d) z_0^{-d} = \frac{1}{2}, 
    %\eqnlabel{pmfconstraint}
    \\
    & \sum_{d=1}^{d_{\max}} g(d) = 1, \quad g(d) \ge 0.
\end{align}
\end{subequations}

\begin{figure}[t]
   \centering
    \includegraphics[width=0.5\textwidth]{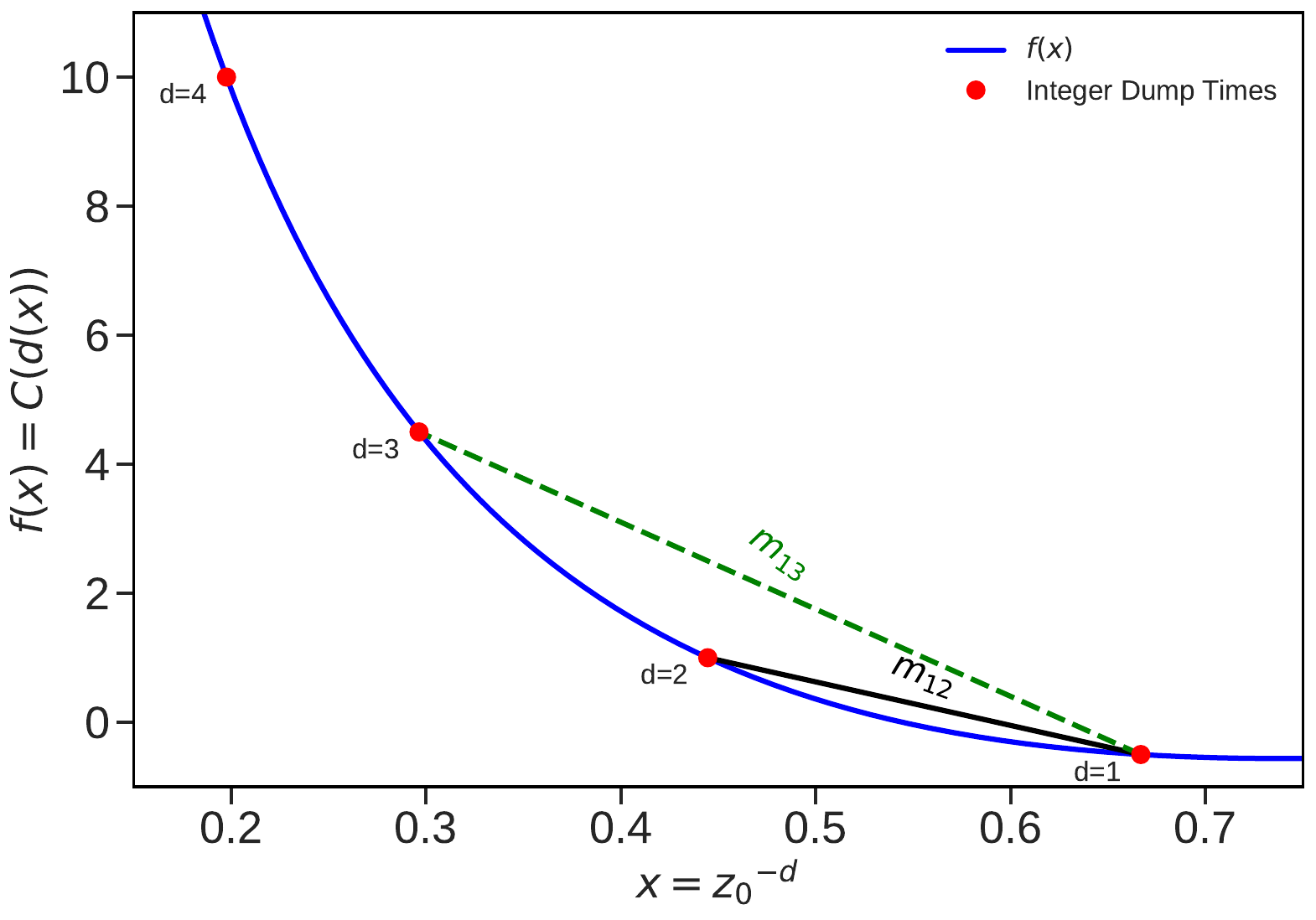}
    \caption{The cost function $f(x)= C(d(x))$ plotted against $x = z_0^{-d}$ for $z_0 = 1.5$ and $\gamma = 1.5$. }
    \label{fvsz}
\end{figure}

We introduce a change of variables from the integer duration $d$ to the real-valued $x = z_0^{-d}$. 
This allows us to to define a continuous cost curve $f(x) \triangleq C(d(x))$, where $d(x) = -\frac{\log x}{\log z_0}$, as illustrated in Fig.~\ref{fvsz}. The original integer durations $d \in \{1, 2, \dots\}$ now correspond to a discrete set of points $\{x_1, x_2, \dots\}$ on this curve, where each 
\begin{align}
    x_i = z_0^{-i} \eqnlabel{xiequals}.
\end{align}
The optimization problem now becomes
\begin{subequations}
\eqnlabel{optnew}
    \begin{align}
    \min_{g(i)} \quad &  \sum_{i=1}^{d_{\max}} g(i) f(x_i) %\eqnlabel{obj}
    \\
    \text{s.t.} \quad & \sum_{i=1}^{d_{\max}} g(i) x_i = \frac{1}{2} ,\eqnlabel{pmfconstraint}
    \\
    & \sum_{i=1}^{d_{\max}} g(i) = 1, \quad g(i) \ge 0 \eqnlabel{nonnegativec}.
\end{align}
\end{subequations}

Since the objective and the constraints are linear functions of the variables $g(i)$, \eqnref{optnew} is a primal linear program (LP). From the fundamental theorem of linear programming, we have that the optimal solution to a primal LP always occurs at a vertex (corner point) of the feasible region. For a feasible set defined by two equality constraints, a vertex corresponds to a distribution with at most two non-zero probabilities.

Let the distribution $g$ be non-zero at two points corresponding to indices $i$ and $k$, with $i < k$. From the problem constraints \eqnref{pmfconstraint} and \eqnref{nonnegativec}, we have
\begin{subequations}
\eqnlabel{gprob}
    \begin{align}
    g(i) = \frac{\frac{1}{2} - x_k}{x_i - x_k},\\
    g(k) = \frac{x_i - \frac{1}{2}}{x_i - x_k} .
\end{align}
\end{subequations}
Since $g(i)$ and $g(k)$ must be non-negative,
\begin{align}\eqnlabel{x-constraint-sandwich}
    x_i \ge \frac{1}{2} \ge x_k.
\end{align}
Taking logarithm base 2 across the inequality, it follows from \eqnref{xiequals}
\begin{align}
    i \le \frac{1}{\log_2 (z_0)} \le k\eqnlabel{nonnegative}.
\end{align}
We want to find a $\gamma$ such that $J(\gamma)= 0$. For a general $\gamma$, the problem is to find the indices $i$ and $k$ that minimize the expected cost
    \begin{align}
     \min_{i, k} \quad & g(i)f(x_i) + g(k) f(x_k)\nn
    = \min_{i, k} \quad &\frac{(\frac{1}{2} - x_k) f(x_i) + (x_i - \frac{1}{2})f(x_k)}{x_i - x_k}. \eqnlabel{Jgamma}
\end{align}

\begin{lemma}\label{lemma:convexity}
    For $\gamma < 2 + \frac{2}{\log(z_0)}$, $f(x)$ is strictly convex over $[0, z_0^{-1}]$. 
\end{lemma}
The proof is in Appendix~\ref{app:convexity}.

\noindent \textbf{Claim:} For any $\gamma$ that satisfies Lemma \ref{lemma:convexity}, the optimal policy that minimizes $J(\gamma)$ is supported on two consecutive integers, with indices $i = \left\lfloor \frac{1}{\log_2(z_0)}\right\rfloor$ and $j = \left\lceil\frac{1}{\log_2(z_0)}\right\rceil$. 

\begin{proof}
    The proof follows by contradiction. Assume the optimal policy that minimizes $J(\gamma^*)$ is supported on two non-consecutive integers, $i$ and $k$, where $k > i + 1$. Let $j$ be an integer such that $i< j < k$.  
    The cost of this assumed optimal policy is 
    \begin{align}
        J_{ik}(\gamma) &= \frac{\left(\frac{1}{2} - x_k\right)f(x_i) + \left(x_i - \frac{1}{2}\right)f(x_k)}{x_i - x_k}.
        \eqnlabel{Jik-cost}
    \end{align}
     
Since $i < j < k$ and $x(d)$ is strictly decreasing, we have $x_i > x_j > x_k$. The constraint value $1/2$ from \eqnref{x-constraint-sandwich} must lie in either $(x_i,x_j]$ or $(x_j,x_k]$. Let's assume the case where $x_i > 1/2 \ge x_j$ (The other case, $x_j > 1/2 \ge x_k$, follows in the same way by comparing $J_{jk}$ and $J_{ik}$). The condition $x_i > \frac{1}{2} \ge x_j$ allows us to construct a new feasible two-point policy on the consecutive integers $(i, j)$ that satisfies the constraints \eqnref{gprob} and \eqnref{nonnegative}. The cost of this new policy is 
\begin{align}
    J_{ij}(\gamma) &= \frac{\left(\frac{1}{2} - x_j\right)f(x_i) + \left(x_i - \frac{1}{2}\right)f(x_j)}{x_i - x_j}.
    \eqnlabel{Jij-cost}
\end{align}

To compare the costs $J_{ik}$ and $J_{jk}$,  we define the slope of the line segment between the points $i$ and $k$, and the point $j$ and $k$ as follows:
\begin{align}
m_{ik} &\triangleq \frac{f(x_k) - f(x_i)}{x_k - x_i},\\
    m_{ij} &\triangleq \frac{f(x_j) - f(x_i)}{x_j - x_i}.
\end{align}
We re-write the cost expressions $J_{ij}(\gamma)$ and $J_{ik}(\gamma)$ in terms of these slopes. Starting with $J_{ik}$ in \eqnref{Jik-cost}, we obtain
\begin{align}
    J_{ik}(\gamma) &= \frac{\left(\frac{1}{2} - x_k\right)f(x_i) + \left(x_i - \frac{1}{2}\right)f(x_k)}{x_i - x_k}\\
    &= \frac{f(x_i) (x_i - x_k) + \frac{1}{2} f(x_i) + x_i f(x_k) - \frac{1}{2} f(x_k) - x_i f(x_i)}{x_i - x_k}\\
    & = f(x_i) + \frac{\left(\frac{1}{2} - x_i\right)(f(x_i) - f(x_k))}{x_i - x_k}\\
    &= f(x_i) + \left(\frac{1}{2} - x_i\right) m_{ik}.
\end{align}

Similarly we can write $J_{ij}$ in \eqnref{Jij-cost} in terms of $m_{ij}$, yielding
\begin{align}
    J_{ij}(\gamma) &=f(x_i) + \left(\frac{1}{2} - x_i\right) m_{ij}.
\end{align}

Strict convexity of the curve $f(x)$ implies $m_{ij} > m_{ik}$.  Also  $\left(\frac{1}{2} - x_i\right)$ is negative. Thus $J_{ij}(\gamma) < J_{ik} (\gamma) $. This contradicts our initial hypothesis that the policy on $(i, k)$ was optimal. The same logic applies if  $x_j > 1/2 \ge x_k$. Therefore the optimal policy cannot be supported on non-consecutive integers. 
\end{proof}

We have proven that the optimal policy must be supported on two consecutive integers, $(i, i + 1)$. It then follows  from \eqnref{nonnegative} that
$i \le \frac{1}{\Lambda} < i + 1$. This uniquely determines the optimal support integers as
\begin{equation}\eqnlabel{ijsandwich}
    i = \left\lfloor \frac{1}{\Lambda}\right\rfloor, \quad
    j = \left\lceil \frac{1}{\Lambda}\right\rceil.
\end{equation}

If $1/\Lambda$ is an integer, then $i = j =1/\Lambda$ and the policy becomes the DAD policy. 

Now, we must find the specific value $\gamma^*$ that satisfies the Dinkelbach condition $J(\gamma^*) = 0$. From the preceding proof, we know that for a given $\gamma$, the optimal policy that minimizes $J(\gamma)$ is supported on two consecutive integers, $i = \lfloor 1/\Lambda \rfloor$ and $j = \lceil 1/\Lambda \rceil$. Therefore, we need to find $\gamma^*$ such that
\begin{align}
    J(\gamma^*) = \min_g \left\{ \E{D^2} - \gamma^* \E{D} \right\} = g(i)(i^2 - \gamma^* i) + g(j)(j^2 - \gamma^* j) = 0,
\end{align}
where $g(i)$ and $g(j)$ are given by \eqnref{gprob}.
Rearranging the terms to solve for $\gamma^*$, we get
\begin{align}
    \gamma^* &= \frac{g(i)i^2 + g(j)j^2}{g(i)i + g(j)j}.
\end{align}
Finally, we need to show that this optimal $\gamma^*$ satisfies the convexity condition from Lemma \ref{lemma:convexity}. 
We write
\begin{align}
    \gamma^* - i &= \frac{g(i)i^2 + g(j)j^2}{g(i)i + g(j)j} - i = \frac{g(j)j(j - i)}{\E{D}}.
\end{align}
Since $g(j) > 0$, $j > 0$, and $j > i$, $\gamma^* - i > 0$, which implies $\gamma^* > i$. 
Similarly, 
\begin{align}
    \gamma^* - j &= \frac{g(i)i^2 + g(j)j^2}{g(i)i + g(j)j} - j 
= \frac{g(i)i(i - j)}{\E{D}}.
\end{align}
Since $g(i) > 0$, $i \ge 1$, and $i < j$,  $\gamma^* - j < 0$, implying $\gamma^* < j$. Thus, we have shown that $i < \gamma^* < i+1$. This establishes that the optimal value of the Dinkelbach variable $\gamma^*$ satisfies
\begin{align}
    \left\lfloor \frac{1}{\log(z_0)}\right\rfloor \le \gamma^* 
    <  \left\lceil \frac{1}{\log(z_0)}\right\rceil.
\end{align}
This value of $\gamma^*$ is guaranteed to satisfy the convexity condition in Lemma \ref{lemma:convexity}. 
The optimal policy is therefore uniquely determined by the probabilities $g(i)$ and $g(j)$ on the support points $i = \lfloor 1/\Lambda \rfloor$ and $j = \lceil 1/\Lambda \rceil$.

To align with the parameterization of our other policies, we define $\tau \triangleq \E{D}$ as the mean inter-dump time. Since $\tau$ must lie between the support integers, we have $i = \lfloor \tau \rfloor$ and $j = \lceil \tau \rceil$.  

Let $p_j$ be the probability of dumping at the later time $j$, and $p_i = 1-p_j$ be the probability of dumping at the earlier time $i$. The mean inter-dump time is 
\begin{align}
    \tau = \E{D} = i  p_i + j  p_j = i(1-p_j) + (i+1)p_j = i + p_j.
\end{align}
This gives an expression for the probabilities in terms of the mean $\tau$, as follows
\begin{align}
    p_j = \tau - i = \tau - \lfloor \tau \rfloor \quad \text{and} \quad p_i =1-(\tau-\floor{\tau}).
\end{align}
We now state the specific age and leakage for this optimal policy.

From Theorem~\ref{thm:RADleakagerate}, the D-DAD policy has  asymptotic MaxL rate 
\begin{align}
    p_i z_0^{-i} + p_j z_0^{-j} = \frac{1}{2},
\end{align}
which has a solution with $0<p_i\le 1$ if \eqnref{ijsandwich} is satisfied.

The average age for the RAD policy in Theorem~\ref{thm: RAD age} can also be expressed in terms of the mean and variance of the inter-dump time $D$ as 
\begin{align}
    \E{A_{m}(n)} &= \frac{1}{\lambda}+ \frac{\E{D}}{2} + \frac{\mathrm{Var}(D)}{2\E{D}} + \frac{1}{2} \eqnlabel{Radnewage}.
\end{align}

We substitute $\E{D} = \tau$. The variance for this two-point distribution is $\mathrm{Var}(D) = p_i p_j$.

Plugging these into \eqnref{Radnewage} yields the average age for the D-DAD policy as
\begin{align}
    \E{A_m(n)}
    &= \frac{1}{\lambda} + \frac{\tau}{2} + \frac{p_i p_j}{2 \tau} + \frac{1}{2}.
\end{align}

\section{Numerical results}
\label{results}

In this section, we present numerical results to illustrate and compare the age-leakage trade-offs for the service and dump policies that we have analyzed. While the previous sections derived the mathematical trade-offs in terms of the leakage rate $\Lambda$ (in bits/slot), we now re-frame the results using its reciprocal, which we define as the average leakage time:
\begin{align}
    \TMaxL \triangleq \frac{1}{\Lambda}.
\end{align}
This metric, $\TMaxL$, represents the average number of slots required to leak one bit of timing information. This allows us to compare timeliness (age) and privacy-efficiency (leakage time) in the same units (slots). A policy that leaks less timing information about the source will have a larger $\TMaxL$.

We will plot our results in the $(\Delta, \TMaxL)$ space. To benchmark these trade-offs, we first define a baseline based on the  \textit{zero-delay server} which transmits each arriving update immediately in the slot that it becomes available for transmission. We note the FCFS system, the LCFS system, and the dumping server all become zero-delay servers when $g(1)=1$, or equivalently when $\tau=1$.  The zero-delay server policy has an age of $\Delta_1 = 1 + 1/\lambda$, a leakage rate of $\Lambda=1$ bit/slot, and thus a leakage time of $\TMaxL=1$ slot. All non-zero-delay servers must sacrifice timeliness with age $\Delta > \Delta_1$ to gain privacy-efficiency with $\TMaxL > 1$.
This framework allows us to define the key performance metric for any policy, its \textit{Server Efficiency}, denoted by $\eta$. This is the slope of the line connecting the baseline point $(\Delta_1, 1)$ to any other point $(\Delta, \TMaxL)$ on the trade-off curve and is given by
\begin{align}
    \eta \triangleq  \frac{\TMaxL - 1}{\Delta - \Delta_1}.\footnotemark
\end{align}
\footnotetext{Policies like DAD are defined only for integer parameters, resulting in a discrete set of achievable $(\Delta, \TMaxL)$ points.}
A policy with a steeper slope is more efficient, as it provides a greater privacy gain for each unit of age sacrificed. 
For practical systems, where a single time slot might be 1 ms or less, the relevant operational regime could involve  average service times of $\tau = 100$ slots or perhaps even higher. We now proceed to compare our policies using this framework.

Our numerical evaluations will also include comparisons with an  FCFS benchmark system. In particular, the  leakage characteristic of the FCFS system is specified by Theorem \ref{thm:max_leakageFCFSLCFS} and is identical to the preemptive LCFS policy  and AoI analysis of these systems is summarized in the following subsection. 

\subsection{FCFS Benchmark Policy}
\label{sec:FCFSbenchmark}
The average AoI for the FCFS Ber/G/1 queue was studied by Tripathi et al.~\cite{tripathi2019age}:
\begin{align}
%\eqnlabel{FCFS-tripathi}
	\Delta_{\mathrm{FCFS}}(g)= 1+ \E{S}
		+ \frac{\lambar(1-\lambda\E{S})}{\lambda M_g(\lambar)}
		+ \frac{\lambda(\E{S^2}-\E{S})}{2(1-\lambda\E{S})} \eqnlabel{FCFStripathi}, 
\end{align}
where
\begin{align}
    M_g(\lambar) &= \sum_{s=1}^\infty g(s) \lambar^s.
\end{align}
A specific instance of this class that is robust to low service rates is the \textit{Memoryless with Bernoulli Thinning} (MBT) policy, which we introduced in our prior work~\cite{NityaPrivacy}. In this policy, the server admits each arriving update into the FCFS queue with probability $\alpha$ and admitted arrivals have memoryless (i.e. geometric) service times. The thinning results in an effective arrival rate of $\lambda_{\mathrm{eff}} = \alpha \lambda$ that ensures queue stability ($\lambda_{\mathrm{eff}} \E{S} < 1$) even when high privacy requires a low service rate. Since the service time is geometric with parameter $\mu$, the age is given by {\cite[Eq.~10]{NityaPrivacy}
\begin{align}
    \Delta_{\mathrm{MBT}}(\alpha, \mu) = \frac{1}{\alpha\lambda} + \frac{1}{\mu} + \frac{\alpha ^ 2\lambda^2 (1-\mu)}{\mu^2(\mu - \alpha \lambda)}.
\end{align}
The thinning strategy can be applied to any service distribution $g(s)$. While \eqnref{FCFStripathi} enables AoI evaluation of any service distribution $g(s)$, the optimal service distribution in the FCFS SMP class remains unresolved because of the  non-linear dependence of the FCFS age $\Delta_{\mathrm{FCFS}}(g)$ on the distribution $g(s)$. However, our numerical evidence does suggest that the same ``greedy'' distribution from Theorem~\ref{thm:optimal_pmfLCFS} may also be optimal for the FCFS case. Therefore, in Section~\ref{sec:coupledresults} we plot the performance of an FCFS policy employing this greedy service distribution to serve as a benchmark for the FCFS policy class.

\subsection{Performance of Coupled Policies}
\label{sec:coupledresults}

We first validate our analytical finding from Theorem~\ref{thm:optimal_pmfLCFS} that the optimal LCFS coupled policy with service PMF $g(s)$ must have a minimum service time $s_{\min}(g)=1$. Fig.~\ref{optimalsmin} plots the best achievable privacy-age trade-off in the $(\Delta, \TMaxL)$ space for policies constrained to $s_{\min}(g)=1$, $s_{\min}(g)=2$, and $s_{\min}(g)=3$. For each curve, we use the optimal ``greedy'' SMP distribution structure derived in Theorem~\ref{thm:optimal_pmfLCFS}, shifted to start at the respective $s_{\min}$. The corresponding average age is calculated using the exact analytical formula from Corollary~\ref{cor:lcfsage}. The leakage time $\TMaxL = 1/\Lambda$ is derived from the asymptotic leakage rate $\Lambda$. For the $s_{\min}(g)=1$ case, we use the closed-form expression $\Lambda = \log_2(1 + \MLC)$, while for $s_{\min}(g) > 1$, $\Lambda$ is calculated numerically using \eqnref{smpleakage} for $n=10000$. The results show that the $s_{\min}(g)=1$ curve is strictly higher than all others, confirming that any coupled LCFS policy with $s_{\min}(g)>1$ is suboptimal. 

\begin{figure}[t]
    \centering
    \includegraphics[width=0.5\textwidth]{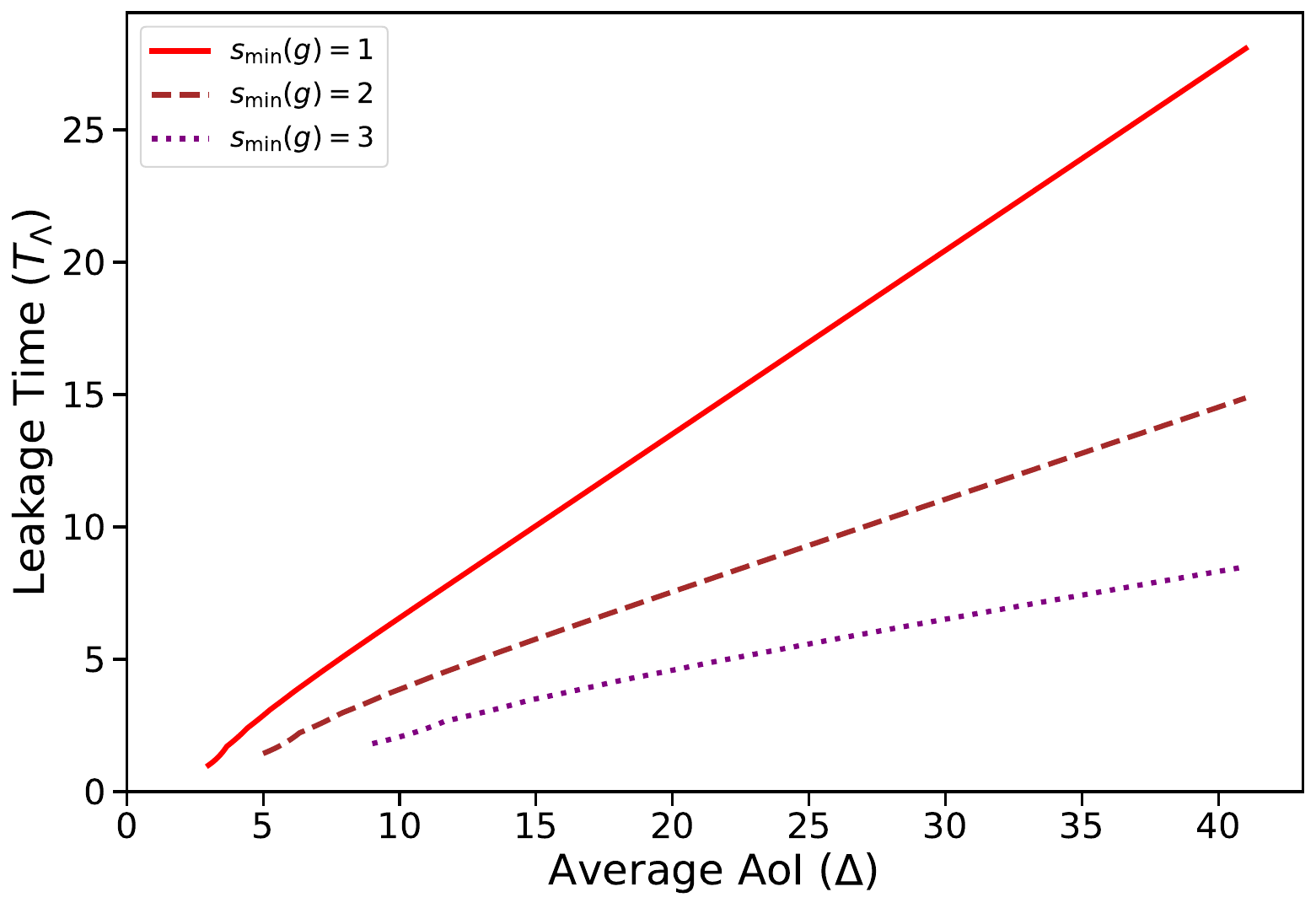}
    \caption{Comparison of the optimal age-leakage trade-off in the $(\Delta, \TMaxL)$ space for LCFS Coupled Policies. Each curve is generated by applying the optimal ``greedy" SMP distribution from Theorem~\ref{thm:optimal_pmfLCFS} for a given leakage constraint $\MLC= g(s_{\min}(g))$. The constraint $\MLC$ is varied in the range $[0.025, 1.0]$ for $s_{\min}(g)=1$, $[0.05, 1.0]$ for $s_{\min}(g)=2$, and $[0.1, 1.0]$ for $s_{\min}(g)=3$. The source arrival rate is $\lambda = 0.5$. Leakage rate for $s_{\min}(g) > 1$ is calculated numerically for $n=10000$ time slots.}
    \label{optimalsmin}
\end{figure}

Fig.~\ref{queueing} compares the age-leakage trade-offs for the FCFS and LCFS coupled policies in the $(\Delta, \TMaxL)$ space, for a source arrival rate of $\lambda = 0.5$. The plot evaluates five distinct policies: the LCFS Greedy (Optimal SMP) and LCFS (Geo) policies, two ``thinned" FCFS policies: MBT and FCFS Greedy with Bernoulli thinning (FCFS GBT), and the unthinned  FCFS Greedy policy. For the thinned policies, the server admits updates into the FCFS queue with probability $\alpha$. For every point on their respective curves (fixed service parameters $\mu$ or $\MLC$), the age is calculated by numerically optimizing $\alpha$ to minimize the age expression~\eqnref{FCFStripathi}. 

For both FCFS and LCFS, we plot the performance of a policy with geometric service times against  the greedy service time distribution, which is shown in Theorem \ref{thm:optimal_pmfLCFS} to be optimal for the SMP class of coupled LCFS systems. The optimal curve is plotted by varying the leakage constraint $\MLC = g(1)$. For each value of $\MLC$, the MaxL rate is determined as $\Lambda = \log(1 + \MLC)$, and the corresponding minimum age is computed using the LCFS age formula from Corollary~\ref{cor:lcfsage} with the uniquely determined greedy PMF. In the geometric LCFS case, varying the service rate $\mu=1/\tau$ controls the operating point on its curve, with the MaxL rate given by $\Lambda = \log(1 + 1/\tau)$ and the age by $\Delta = 1/\lambda + \tau$. The leakage rate for each FCFS curve is generated similarly though the age $\Delta_{\mathrm{FCFS}}(g)$ is calculated using \eqnref{FCFStripathi}.
\begin{figure}[t]
    \centering
    \includegraphics[width=0.5\textwidth]{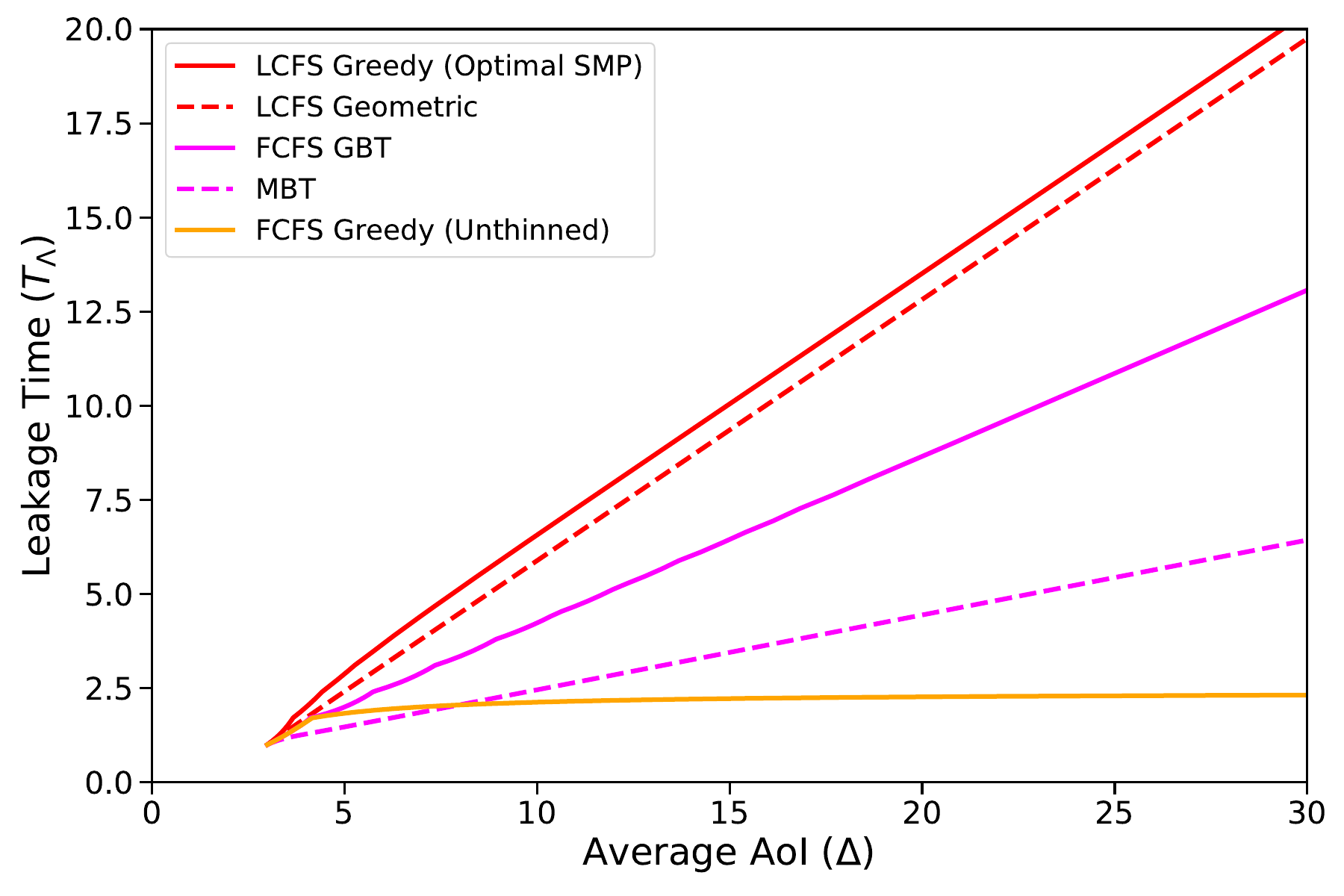}
    \caption{Leakage time as a function of the average AoI for Coupled Policies with $\lambda = 0.5$. The plot compares LCFS policies against thinned and unthinned FCFS policies. For the greedy policies, the leakage constraint $\MLC$ is varied from 0.1 to 1. For the geometric policies, the service rate $\mu=1/\tau$ is varied from 0.1 to 1. The thinned policies optimize the admission probability $\alpha$ at every point on the curve to minimize the age for that specific service rate. }
    \label{queueing}
\end{figure}

As shown in the Fig.~\ref{queueing} plot, the left-most point for all policies is the zero-delay server baseline. This point, $(\Delta, \TMaxL) = (3, 1)$, corresponds to the privacy constraint ($\MLC=1$ or $\mu=1$), where the age is $\Delta_1 = 1 + 1/\lambda = 3$ and the leakage time is $\TMaxL=1$. As we sacrifice timeliness (age increases to the right), a more efficient policy will have a higher leakage time $\TMaxL$. The results 
demonstrate the significant advantage of the LCFS discipline. On this $(\Delta, \TMaxL)$ plot, the LCFS curves are strictly higher than the FCFS curves. This shows that LCFS policies are far more efficient. For any given age $\Delta$, they provide a much greater leakage time $\TMaxL$. An analysis of the asymptotic slopes of the $(\Delta, \TMaxL)$ curve in the  high age regime, as  $\Delta\to \infty$ reveals a key difference between the LCFS and FCFS policies. Both the LCFS policies, have an asymptotic slope of $\ln(2)$ indicating that increasing the AoI by one time unit increases the leakage time by $\ln(2)$ units of time. As the performance difference between the LCFS policies is a constant age offset, they can be interpreted as having equal efficiency in the high-age regime. The LCFS service prioritizes timeliness by serving the newest update and discarding older ones, which is an effective strategy for minimizing age. FCFS, in contrast, must process every update, including stale ones, in their arrival order, which is detrimental to age. 

The thinned policies improve upon unthinned FCFS Greedy by maintaining stability. The thinned policies have non-zero efficiency  but they are not as good as LCFS. For the unthinned FCFS Greedy policy, as $\Delta\to \infty$, the asymptotic slope converges to $0$. This shows that this policy provides zero additional  privacy despite large sacrifices in timeliness.

\subsection{Performance of Decoupled Policies}

If all updates need not be delivered to the monitor, policies can drop stale updates to prioritize transmission. Since we can design the service (or ``dump") distribution, a key question is whether a deterministic or random schedule is better for minimizing age. Fig.~\ref{RandomDeterministic} compares the performance of the DAD policy against two RAD policies: one using uniform dumps and another using geometric dumps for various source arrival rates $\lambda$. It is important to note that the LCFS policy with geometric service is functionally equivalent to the RAD policy with geometric dumps.~\footnote{Throughout this section, we refer to a RAD policy where the inter-dump times $D_k$ follow a Geometric distribution as ``Geometric dumps'' (or ``RAD Geometric''), and similarly for ``Uniform dumps'' when $D_k$ is Uniform.} In both policies, an available update is transmitted in any given slot with a constant probability, leading to identical age-leakage performance curves. 

\begin{figure}[t]
    \centering
    \includegraphics[width=0.5\textwidth]{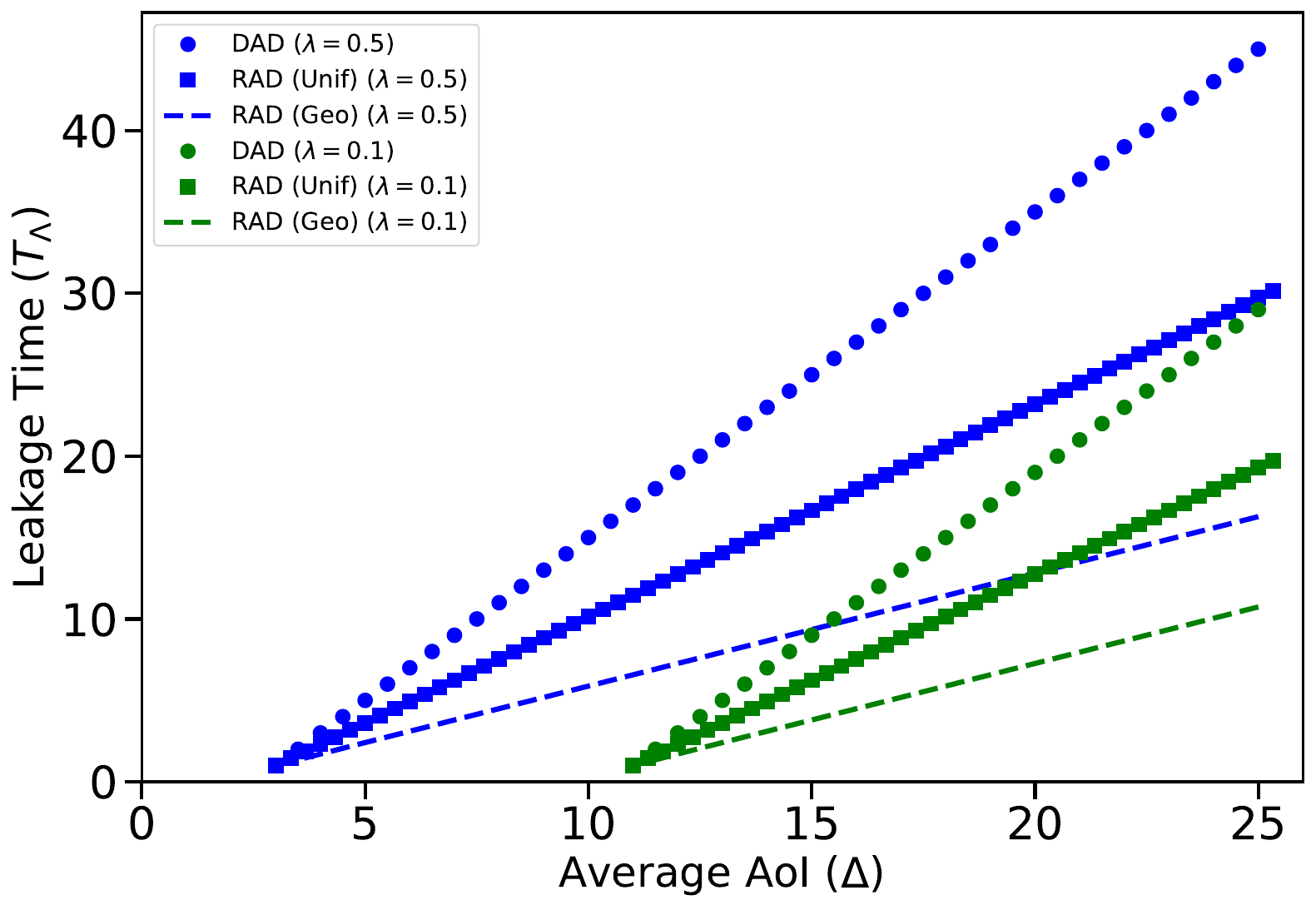}
    \caption{Leakage time as a function of the average AoI for three decoupled policies: DAD, RAD with uniform dumps, and RAD with geometric dumps. Each policy is plotted for arrival rates $\lambda = 0.1$ and $\lambda = 0.5$. All three policies are parameterized by their mean inter-dump time $\tau= \E{D}$. The geometric RAD policy is plotted by varying $\tau$ continuously. The DAD and uniform RAD policies are plotted by varying their underlying integer parameters, which restricts their mean $\tau$ to discrete integer (for DAD) or half-integer (for uniform) values. }
\label{RandomDeterministic}
\end{figure}

For a given age, the DAD policy achieves the highest leakage time $\TMaxL$. To understand this trade-off, we recall from Theorem~\ref{thm: DAD-leakage} and Corollary~\ref{Cor: Special-RAD-age} that both the leakage time and age for all RAD policies (including the optimal D-DAD policy) are parameterized by the deterministic dump period $\tau$. Specifically, for the DAD policy, $\TMaxL = 1/\Lambda = \tau$ and $\Delta = 1/\lambda + (\tau+1)/2$. As $\tau$ increases, the age and leakage time both increase. The DAD policy is followed by the RAD policy with uniform dumps which has leakage rate  given by Corollary~\ref{cor:RADunifleakage} and $\Delta = 1/\lambda + (2\tau+1)/3$, and  the geometric RAD policy performs the worst, with leakage time $\TMaxL = 1/\log(1+1/\tau)$ and $\Delta = 1/\lambda + \tau$. This performance gap stems from the variance of the dump schedule, as shown by the general RAD formula in \eqnref{Radnewage}. The zero-variance DAD policy is optimal for integer-valued dump periods $\tau$, while the low-variance uniform policy is nearly as good. The high-variance geometric policy is significantly worse for age.

Fig.~\ref{RandomDeterministic} also illustrates that for a fixed value of the age $\Delta$, the leakage time $\TMaxL$ is increasing in the source rate $\lambda$. For a given total age, a more active source has a smaller source age penalty ($1/\lambda$) as seen from \eqnref{Radnewage}. This allows the server to use a longer, more private dump period $\tau$, to achieve the same total age, resulting in a higher $\TMaxL$.

\subsection{Comparison between Coupled and Decoupled Policies}
 Fig. \ref{DADLCFS} compares the best-performing policy from each class: The DAD policy vs.~the optimal LCFS coupled policy. The discrete operating points of the DAD policy are shown as blue circles on the plot, corresponding to integer dump periods $\tau$. The solid black curve connects these points, representing the continuous trade-off achieved by the optimal D-DAD policy. It is important to note that this connecting path is not a straight line, a property that becomes evident in a magnified view of any segment. The observed curvature arises because both the average age and the asymptotic leakage rate are non-linear functions of the dithering probability. 
 \begin{figure}[t]
    \centering    \includegraphics[width=0.5\textwidth]{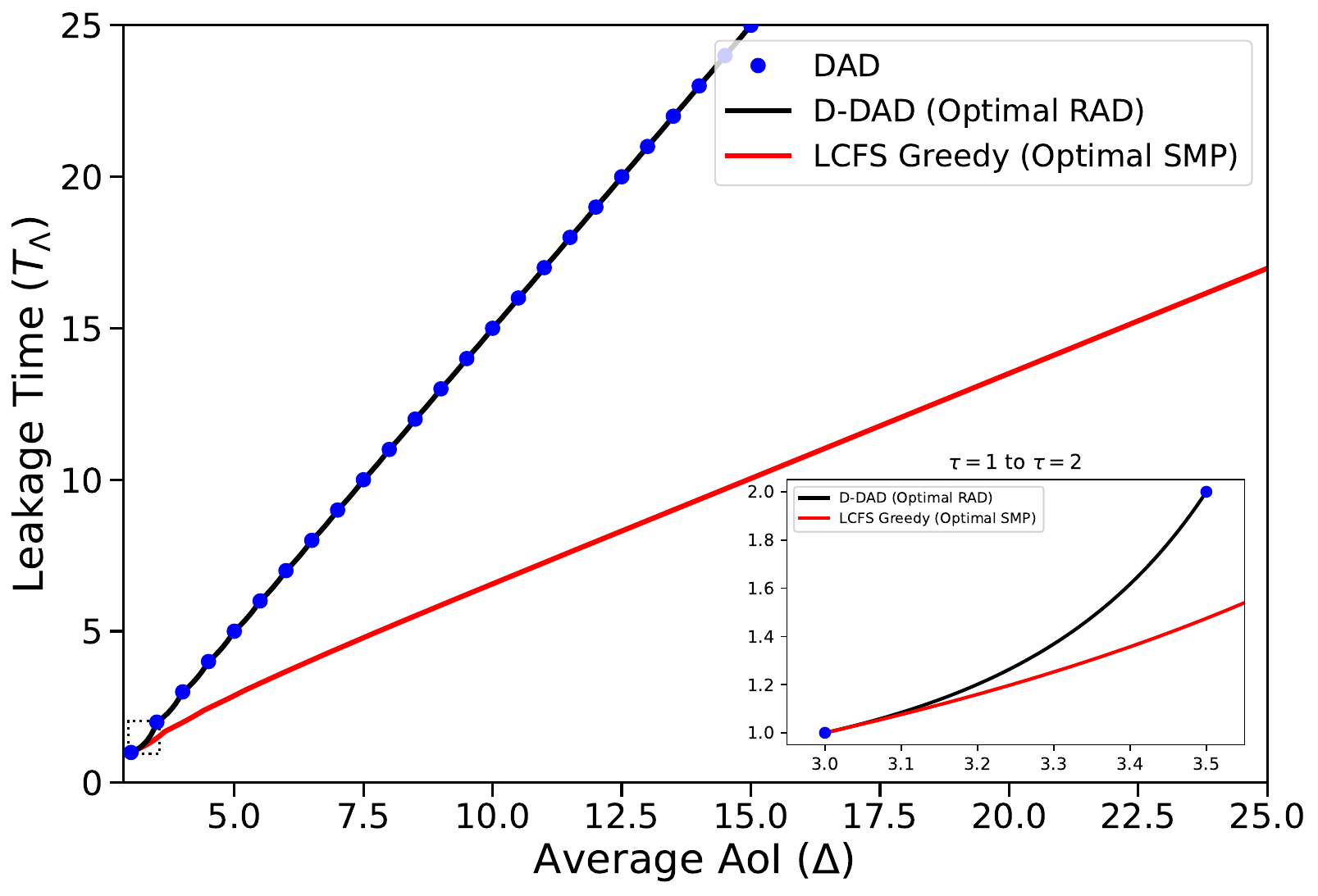}
    \caption{Leakage time as a function of the average AoI time for the DAD and  SMP classes for  $\lambda=0.5$. For the optimal LCFS,  $\MLC$, which specifies the leakage  is varied from 0.01 to 1. For DAD,  $\tau$ varies from $1$ to $25$.}
    \label{DADLCFS}
\end{figure}

For any given age $\Delta$, the D-DAD policy provides a strictly higher leakage time $\TMaxL$ than the optimal LCFS Greedy policy. This shows that the decoupled class is more efficient. 

The fundamental advantage of the D-DAD policy is revealed in the asymptotic limit as $\Delta\to \infty$. The asymptotic slope for the greedy policy converges to $\ln(2)$, whereas the asymptotic limit for the D-DAD policy converges to $2$. 
%\ns{Efficiency does not drop below 0.693 as $\Delta$ approaches $\infty$. Prove this. }
This shows that the D-DAD policy is far superior in the high-age regime, providing significantly more privacy-efficiency for any given age. This regime is relevant for applications where timeliness is measured on a slower time scale than individual packet delivery times. 

This highlights a fundamental advantage of decoupling the transmission schedule from the arrival process. The optimal LCFS SMP is a coupled policy; its decision to transmit is directly triggered by an update's arrival, committing the server for a specific service time. In contrast, the D-DAD policy is decoupled. It follows its own deterministic schedule, which allows it to optimally dither between waiting to accumulate an even fresher update and dumping the freshest available one at the moment that balances the age-leakage trade-off. 

\subsection{Results for Markovian sources}

Our results on leakage require only that the distribution on $X^n$ has full support: in that case, we can identify the maximum likelihood input for every output. The age analysis, however, depends on the actual distribution of the arrivals $X^n$. One of the simplest non-i.i.d.~arrival distributions with full support is a Markov source, which we can use to model bursty traffic.
We present simulation results for our main coupled (FCFS, LCFS)  and decoupled (DAD) policies under a two-state discrete-time Markov source. The state in each time slot determines the source's activity. 
\begin{itemize}
    \item State 1 (Active): The source generates an update
    \item State 0 (Inactive): The source does not generate an update.
\end{itemize}
With $P_{ij}$ denoting the probability of a state transition from $i$ to $j$, the effective arrival rate $\lambda_{\mathrm{eff}}$ is defined as the steady-state probability of the source being in the active state and is given by
\begin{align}
    \lambda_{\mathrm{eff}} = \frac{P_{01}}{P_{01} + P_{10}}.
\end{align}
For each scenario, we compare an i.i.d.~Bernoulli source against a bursty two-state Markov source with transition matrix $P$ that has the same effective rate:
\begin{itemize}
    \item Low-rate ($\lambda_{\mathrm{eff}}= 0.2$): The Markov source ($P_{01}=0.05, P_{10}=0.2$) is characterized by long periods of silence (average inactive period of $1/P_{01} = 20$ slots) and short bursts of activity (average active period of $1/ P_{10} = 5$ slots).
    \item High-rate ($\lambda_{\mathrm{eff}}= 0.8$): The Markov source ($P_{01}=0.2, P_{10}=0.05$) is characterized by short periods of silence (average inactive period of $1/P_{01} = 5$ slots) and long, sustained bursts of activity (average active period of $1/ P_{10} = 20$ slots).
\end{itemize}
The results for the low-rate scenario are shown in Fig.~\ref{fig:iidbursty_TMaxL}, and for the high-rate scenario in Fig.~\ref{fig:iidburstylamb0.8_TMaxL}. For each policy, the leakage rate is calculated using 
%the analytical formulas derived in 
Theorem~\ref{thm:max_leakageFCFSLCFS} for the FCFS and LCFS policies and Theorem~\ref{thm: DAD-leakage} for the DAD policy, as the Markov source still has full support. The average age for the FCFS policy is found via simulation ($5\times 10 ^5$ slots). 
For the LCFS and DAD policies under Markovian arrivals, we again use the sampling of age processes approach to get $\E{A_m} = \E{A_s} + \E{Z_m}$ as given in \eqnref{samplingeqn}. The LCFS policy with geometric service is functionally equivalent to the RAD policy with geometric dumps. Using the exact analytical formula for the age at the input to the server $\E{A_s}$ derived in Appendix~\ref{app:markov_age}, and the server sampling age $\E{Z_m}$ from Corollary~\ref{Cor: Special-RAD-age}, the average age at the monitor is
\begin{align}
    \E{A_m} = 1 + \frac{P_{10}}{P_{01}(P_{01}+P_{10})} + \E{Z_m},
\end{align}
where $\E{Z_m} = \tau$ for the LCFS policy (geometric service with mean $\tau$) and $\E{Z_m} = (\tau+1)/2$ for the DAD policy (deterministic dumps with period $\tau$).

\begin{figure}[t]
    \centering
    \includegraphics[width=0.5\textwidth]{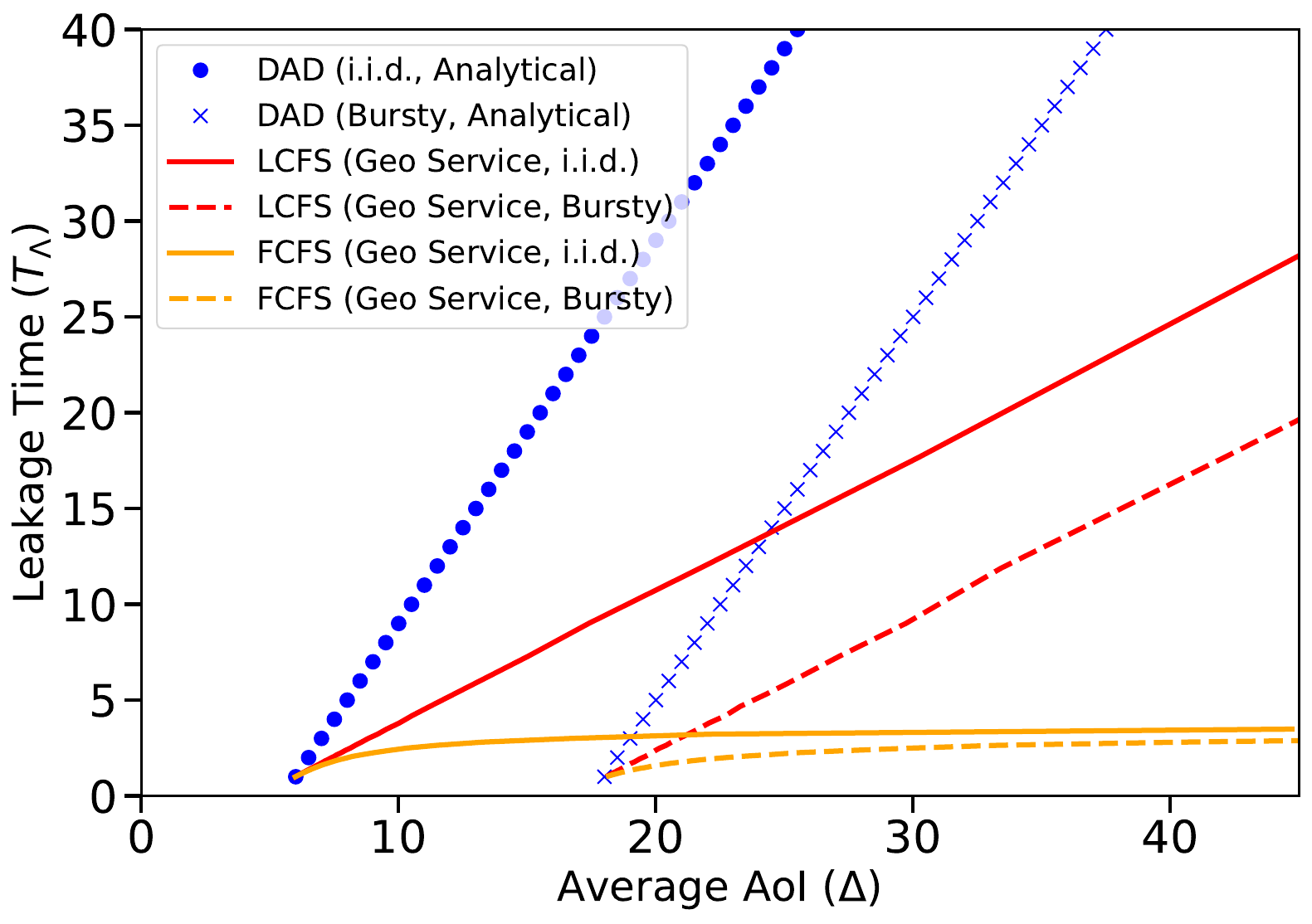}
    \caption{Leakage time as a function of the average AoI for FCFS, LCFS, and DAD policies under i.i.d.~Bernoulli and bursty Markov traffic ($P_{01}=0.05, P_{10}=0.2$). Both sources have $\lambda_{\mathrm{eff}}= 0.2$. The FCFS and LCFS curves are generated by varying the geometric service rate $\mu = 1/\tau$. The DAD curves are generated by varying the deterministic dump period $\tau$ from $1$ to $40$.}
    \label{fig:iidbursty_TMaxL}
\end{figure}

\begin{figure}[t]
    \centering
    \includegraphics[width=0.5\textwidth]{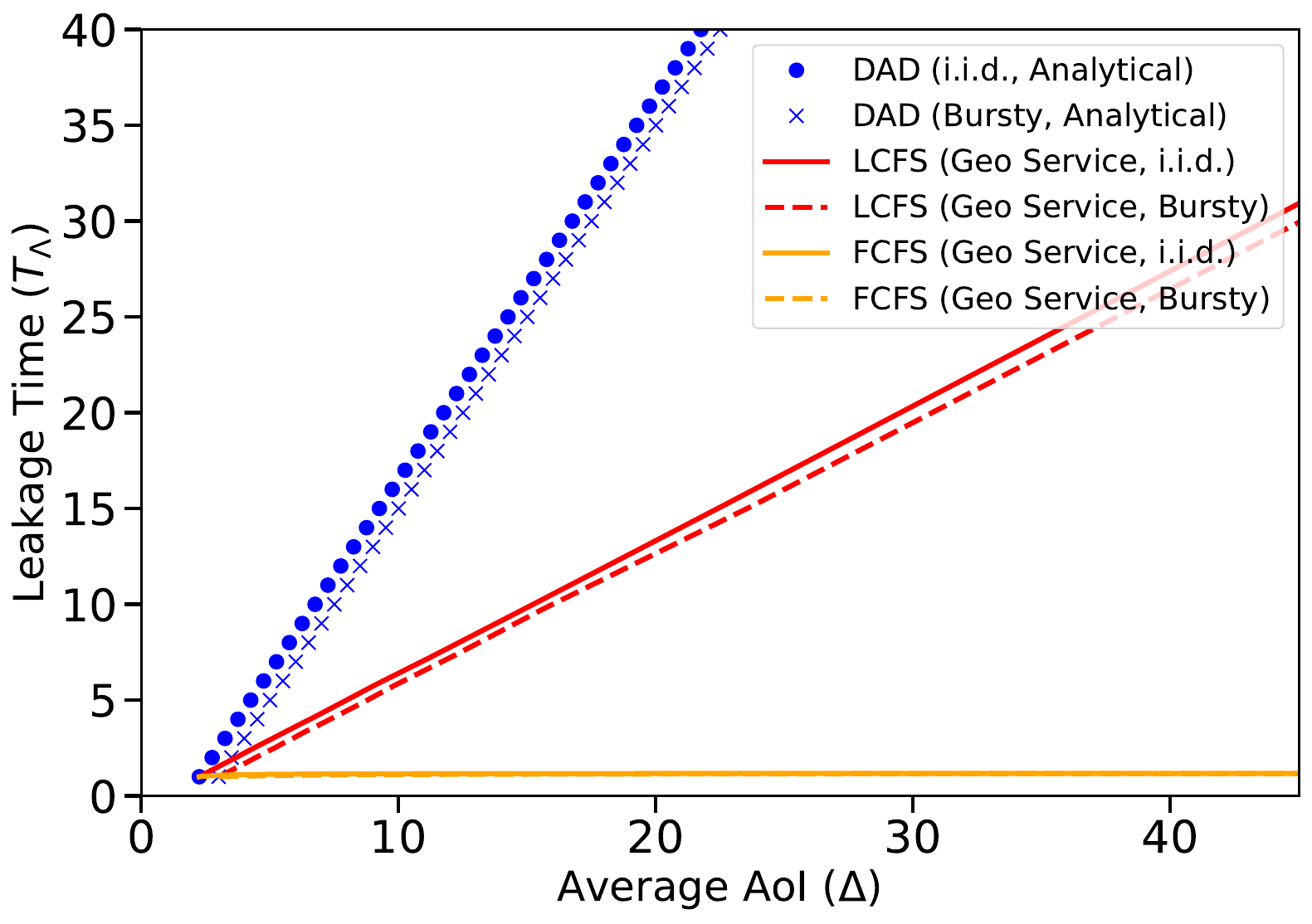}
    \caption{Leakage time as a function of the average AoI for FCFS, LCFS, and DAD policies under i.i.d.~Bernoulli and bursty Markov traffic ($P_{01}=0.2, P_{10}=0.05$). Both sources have $\lambda_{\mathrm{eff}}= 0.8$. The FCFS and LCFS curves are generated by varying the geometric service rate $\mu = 1/\tau$. The DAD curves are generated by varying the deterministic dump period $\tau$ from $1$ to $40$.}
    \label{fig:iidburstylamb0.8_TMaxL}
\end{figure}

In the low-rate scenario, the performance of all policies degrades under bursty traffic, shifting the curves to the right, and have higher age for the same leakage time. For the DAD and LCFS policies, this shift is a constant. This gap is exactly the difference in the average age at the input of the server $\E{A_s}$ between the bursty and i.i.d.~sources, representing the additional staleness caused by the source's long inactive periods. However, a key insight from the plot is that the asymptotic slopes for the bursty and i.i.d.~cases converge to the same values in the high-age regime. This indicates that the privacy efficiency of the server policies is determined by their structure (coupled vs.~decoupled), while source burstiness just contributes an additive age penalty.

In the high-rate scenario, the difference between i.i.d.~and bursty traffic becomes negligible for the LCFS and DAD policies. These policies discard stale updates and are hence resilient to the long bursts of activity. Their slopes are identical to the i.i.d.~case. In contrast, the FCFS policy suffers significantly because the long bursts cause the  queue to become backlogged, making the age large.

\subsection{Impact of Full-Support Assumption}

Our analysis has focused on sources with full support (FS). As our source is a Bernoulli process, any input sequence $x^n\in\{0,1\}^n$ is possible. This means that the adversary cannot rule out any input sequence a~priori. Thus the service policy must protect the user's privacy without relying on any assumed structure in the user's behavior. Moreover, this assumption permits a broad class of inputs for which we can rigorously analyze MaxL and draw general conclusions regarding the trade-offs offered by different server policies. 

In contrast, for sources with restricted support (RS), the input sequences are constrained to a subset of the entire input space, $\mathcal{X}_\mathrm{RS}\subset \{0,1\}^n $. Examples include sources with a minimum inter-arrival time or sources that can transmit only in specific slots (e.g., even time slots). For any fixed service policy, a direct consequence of the maximum likelihood and summation operations in the MaxL definition \eqnref{eqn:maxL_def} being performed over smaller sets is  that the calculated leakage for an  RS source is less than or equal to that of a FS source. An extreme instance is a deterministic source, for which the MaxL is 0. This shows that MaxL measures the reduction in uncertainty. When the source is deterministic, the adversary's a~priori guessing probability is already 1. Observing the output sequence, cannot improve the guessing probability, and thus no new information is leaked. Thus, prior information about the source's structure leaves less information to leak, and our analysis for FS sources serves as a worst-case privacy guarantee. 

Optimizing the age-leakage trade-off for RS sources can be challenging. The class of restricted sources is ill-defined, as there are many ways an arrival process can be constrained (e.g., specific periodicities, or forbidden patterns). Also, exploiting these structural restrictions would require the server to employ source-specific policies, for which evaluating both leakage and age becomes more difficult compared to the policies we analyze in this paper.

\section{Conclusion and Future Work}\label{conclusion}

In this paper we investigated the trade-off between AoI and MaxL in a discrete-time status updating system. We considered a model with Bernoulli source arrivals and an eavesdropping adversary and analyzed the performance of two broad classes of server policies: coupled policies, where the service is tied to arrivals, and decoupled policies, where service is based on an independent timer. For the coupled class, we analyzed FCFS as a benchmark and showed that an LCFS policy with a greedy service time distribution is optimal within the class of SMP policies. 

For the decoupled class, we characterize the structure of the optimal policy as dithering between two adjacent deterministic dumps. We also find that the optimal decoupled policy (D-DAD) achieves a better age-leakage trade-off than the optimal coupled policy (LCFS greedy). The optimal strategy in both the classes is a dithering policy. The D-DAD policy dithers between two deterministic periods that differ by one, while the LCFS greedy policy dithers between two discrete-uniform distributions whose support differs by one. 

This work opens several avenues for future research. We limited our analysis to Bernoulli arrival processes that simplify the leakage analysis because all input sequences $x^n$ are in the support set. Future work should explore sources with memory, such as Markovian arrival processes or run-length limited constraints. Such models would not have full support, which would change the structure of the maximum likelihood input sequence and alter the age-leakage trade-off. It remains to be seen if a server efficiency of greater than $2$ can be achieved. 

We also limited our attention to a system with a single source. Extending the analysis to a system with multiple sources competing for a single server would be an interesting problem. This introduces  scheduling problems where the server must decide which source to update to manage the system-wide balance of age and privacy. 

Many real-world sensor networks are resource constrained. Incorporating factors such as energy harvesting at the source or server would create  more complex trade-offs between age, leakage, and energy consumption. 
 In this work, we disallowed policies which generate ``dummy'' updates. A compelling extension would be to allow the server to inject dummy traffic but at a certain cost. Allowing $\delta(n)$ dummy updates would be an interesting direction for future work but seems significantly more challenging. This would create a more complicated policy space where the server's actions are state-dependent (e.g., based on the current AoI). While  age-optimality of server policies has been formulated as an MDP~\cite{CeranHARQ, tang2020minimizing, bedewy2021optimal, hsu2018age, sathyavageeswaran2024timely},
 %\RY{cite REFS}
 analyzing the MaxL for such a policy is especially difficult, as the set of achievable output sequences would depend on the complex evolution of the system state.

\section{Appendix}
\subsection{Proof of Lemma ~\ref{maxLlemma}}
\label{Lemma: FCFS Leakage}
    We prove this claim by induction. We start by showing the claim is true for $n = s_1$. There are two possible achievable output sequences: $y^{s_1} = \mathbf{0}^{s_1}$ and $y^{s_1} = (0, \ldots, 0, 1)$.
    
    Consider the case $y^{s_1} = (0, \ldots, 0, 1)$. 
    \begin{subequations}
\eqnlabel{y1givenx11-x10}
    \begin{align}
    \pmf{Y^{s_1}|X^{s_1}}{(0,\ldots,0,1)|(1,0,\ldots,0)} &= g(s_1) \\
    \pmf{Y^{s_1}|X^{s_1}}{(0,\ldots,0,1)|x^{s_1}} &= 0, \quad \text{for any } x^{s_1} \text{ with } x_1=0.
\end{align}
\end{subequations}
It follows from \eqnref{y1givenx11-x10} that
\begin{align}
     \max_{x^{s_1} \in \mathcal{X}^{s_1}} \pmf{Y^{s_1}|X^{s_1}}{(0,\ldots,0,1)|x^{s_1}} = g(s_1). \eqnlabel{basecase1}
\end{align}
Similarly, for $y^{s_1} = \mathbf{0}^{s_1}$, 
\begin{subequations}
\eqnlabel{y10givenx11-x10}
    \begin{align}
    \pmf{Y^{s_1}|X^{s_1}}{\mathbf{0}^{s_1}|\mathbf{0}^{s_1}} &= 1 \\
    \pmf{Y^{s_1}|X^{s_1}}{\mathbf{0}^{s_1}|(1,0,\ldots,0)} &= 1 - g(s_1).
\end{align}
\end{subequations}
It follows from \eqnref{y10givenx11-x10} that
\begin{align}
    \max_{x^{s_1} \in \mathcal{X}^{s_1}} \pmf{Y^{s_1}|X^{s_1}}{\mathbf{0}^{s_1}|x^{s_1}} = 1. \eqnlabel{basecase2}
\end{align}
It follows from \eqnref{basecase1} and \eqnref{basecase2} that
\begin{align}
    \max_{x^{s_1} \in \mathcal{X}^{s_1}} \pmf{Y^{s_1}|X^{s_1}}{y^{s_1}|x^{s_1}} = [g(s_1)]^{\text{wt}(y^{s_1})}, 
\end{align}
and the maximum is achieved by the time-shifted input sequence. Thus the claim holds for $n = s_1$. Now assume the claim holds for all sequence lengths $j < n$. Then for any achievable output sequence $y^j$, where $j < n$, we have
\begin{align}
    \max_{x^j \in \mathcal{X}^j} \pmf{Y^j|X^j}{y^j|x^j} = [g(s_1)]^{\text{wt}(y^j)}, \qquad j < n
\end{align} which is achieved by the time-shifted input sequence $x^j$. 

Now we need to show the claim holds for length $n$. Let $y^n$ be an achievable output sequence. Let $m$ be the time of the last departure before time $n$, defined as
\begin{align}
    m \triangleq \max\{t \colon  t < n, y_t=1\}, 
\end{align} with $m=0$ if no such $t$ exists.  The prefix $y^m$ is also achievable, because $y^n$ is achievable. 
\begin{align}
\spmf{Y^n|X^n}{y^n|x^n} &= \spmf{Y^m|X^m}{y^m|x^m} \spmf{Y_{m+1}^n|Y^m,X^n}{y_{m+1}^n|y^m,x^n}
\end{align}
Since the departure process until time $m$ does not depend on the arrivals after time $m$, we can write
\begin{align}
\max_{x^n \in \support{X}} \pmf{Y^n|X^n}{y^n|x^n} \nonumber \
&= \max_{x^n \in \support{X}} \left[ \spmf{Y^m|X^m}{y^m|x^m}  \spmf{Y_{m+1}^n|Y^m,X^n}{y_{m+1}^n|y^m,x^n} \right] \ \\
&= \left(\max_{x^m \in \mathcal{X}^m} \spmf{Y^m|X^m}{y^m|x^m}\right)  \left(\max_{x_{m+1}^n} \spmf{Y_{m+1}^n|Y^m,X^n}{y_{m+1}^n|y^m,x^n}\right) , \eqnlabel{twomax}
\end{align}
and this maximum is acheived by the time-shifted input sequence $x^m$. 

By the induction hypothesis
\begin{align}
    \max_{x^m \in \mathcal{X}^m} \pmf{Y^m|X^m}{y^m|x^m} = [g(s_1)]^{\mathrm{wt}(y^m)}.
\end{align}
For the second factor on the right side of \eqnref{twomax}, based on the sequence $y_{m+1}^n$ we consider two cases, $Y_n = 0$ and $Y_n = 1$. 
For the case when $Y_n = 0$ ($y_{m+1}^n = (0, 0, \ldots, 0)$ ), we have 
\begin{align}
    \max_{x^n} \spmf{Y_{m+1}^n|Y^m,X^n}{(0,\ldots,0)|y^m,x^n} \le  1,
\end{align}
where the unity upper bound holds trivially. This unity upper bound is achieved when the input for the prefix $y^m$ is the time-shifted input, and the input for the suffix $x_{m + 1}^n $ is the all zero vector. 

We now consider the case when $Y_n = 1$. The suffix has exactly one departure, $y_{m + 1} ^n = (0,0,\ldots, 1)$. This output is achievable only if $n - m \ge  s_1$. To generate this output, an update must arrive at some time $i$ where ($m < i \le n $) to an idle server. The probability of this event is $g(n - i + 1)$. To maximize this probability, the input suffix $x_{m+1}^n$ must be chosen to yield the most probable service time. By our assumptions on the service time distribution, the most probable service time is $s_1$, occurring with probability $g(s_1)$. This is achieved by setting the arrival time to $i=n -s_1+1$.
Hence, it follows that
\begin{align}
    \max_{x_{m+1}^n } \spmf{Y_{m+1}^n|Y^m,X^n}{y_{m+1}^n|y^m,x^n} = g(s_1).
\end{align}
This maximum is achieved by an input suffix $x_{m+1}^n$ that is an all-zero vector except for a single arrival at time $n-s_1+1$. This corresponds exactly to the input constructed by the time-shifting rule for $y^n$. Now returning to \eqnref{twomax} we see that the maximum is achieved by time-shifted input sequence $x^n$. Hence the claim holds for $n$, thereby completing the proof. 

\subsection{Proof of Corollary~\ref{cor:lcfsage}}
\label{app: corollary1}

    When the arrival process is Bernoulli with parameter $\lambda$, we have $\E{X}= \frac{1}{\lambda}$ and $\E{X^2} = \frac{2-\lambda}{\lambda^2}$.   Defining $\lambar \triangleq 1 - \lambda$, we have $P(X \ge s) = \sum_{k = s} ^\infty {(1 - \lambda) }^ {k - 1} \lambda = {\lambar} ^ {s - 1} $.
    
We also have 
\begin{align}
    P(S\le X) &= \sum _{s = 1}^ \infty P(S \le X | S  =s) P (S = s)\\
    &= \sum _{s = 1} ^ \infty P(X \ge  s) P(S = s)\eqnlabel{PSlessX}\\
    &= \sum _{s = 1} ^ \infty {\lambar} ^ {s - 1} P(S = s)= \E{\lambar ^ {S - 1}}. \eqnlabel{Prob1}
\end{align}

Next we compute $\E{\min(X,S)}$. We write
\begin{align}
    \E{\min (X,S)}& = \sum _{k = 1} ^ \infty P(\min(X,S) \ge k) = \sum _{k = 1} ^ \infty P(X \ge k , S \ge k ).
\end{align}
Since $S$ and $X$ are independent,
\begin{align}
    \E{\min (X,S)}& = \sum _{k = 1} ^ \infty P(X\ge k) P(S \ge k )\\
    & = \sum _{k = 1} ^ \infty {\lambar ^ {k - 1} \sum _ {j = k} ^ \infty P(S = j)}\\
     & = \sum _{j = 1} ^ \infty P(S = j) \sum _{k = 1} ^ j {\lambar} ^ {k - 1}\\
     & = \sum _{j = 1} ^ \infty P(S = j) \left( \frac{1 - {\lambar^ j}}{1 - \lambar}\right)\\
     & = \frac{1 - \E{\lambar^ S}}{\lambda} .\eqnlabel{Emin}
\end{align}

It follows from \eqnref{LCFSTripathi}, \eqnref{Prob1} and \eqnref{Emin} that 
\begin{align}
    \Delta_\mathrm{LCFS-Ber/G/1 } &= \frac{1}{\lambda} + \frac{1- \E{\lambar ^ S}  }{\lambda \E{\lambar ^ {S-1}}} 
    = 1 + \frac{1}{\lambda \E{\lambar ^ {S-1}}}.
\end{align}

\subsection{Proof of Theorem~\ref{thm:optimal_pmfLCFS}}
\label{App: Optimal LCFS}
    The proof consists of two parts. First we prove that the optimal PMF must have a minimum service time of $s_1 = 1$. Second, we prove, given $s_1=1$, the PMF must have the ``greedy" structure.
    From Corollary~\ref{cor:lcfsage}, the LCFS AoI $\Delta_\mathrm{LCFS-Ber/G/1 }$, denoted here by the shorthand notation $\Delta(g)$, 
    is a strictly decreasing function of $\E{\lambar ^ {S-1}}$. 

Let $g(s)$ be any valid PMF in $\mathcal{G}_\mathrm{SMP}$ with a minimum service time $s_1 > 1$ and a corresponding leakage constraint $g(s_1) = \MLC$. We construct a new PMF $g'(s)$ by shifting the entire distribution to the left:
\begin{align}
    g'(s) = g(s + s_1 - 1) \qquad \text{for } s = 1, 2, \ldots.
\end{align}
The minimum service time for the new PMF is $1$, and $g'(1) = \MLC$. Since $g(s)$ is an SMP distribution, $g'(1) = \MLC \ge g(s_1 + k) = g(1+k)$ for $k \ge 1 $, so $g'(s)$ also satisfies the SMP property and meets the same leakage constraint. 

Now we compare the age performance by comparing the term $\E{\lambar^{S-1}}$ for both distributions. 
\begin{align}
    \operatorname{E}_{g'}[\lambar^{S-1}] &= \sum_{s=1}^{\infty} g'(s) \lambar^{s-1} = \sum_{s=1}^{\infty} g(s + s_1 - 1) \lambar^{s-1}. \eqnlabel{Egprime}
\end{align}
Let $t = s + s_1 - 1$, then \eqnref{Egprime} becomes
\begin{align}
\operatorname{E}_{g'}[\lambar^{S-1}] &= \sum_{t=s_1}^{\infty} g(t) \lambar^{t - s_1} = \lambar^{1-s_1} \sum_{t=s_1}^{\infty} g(t) \lambar^{t-1} = \lambar^{1-s_1} \operatorname{E}_g[\lambar^{S-1}].
\end{align}
Since we assumed $s_1 > 1$ and we know $0 < \lambar < 1$,  $\lambar^{1-s_1} > 1$. Therefore, $\operatorname{E}_{g'}[\lambar^{S-1}] > \operatorname{E}_g[\lambar^{S-1}]$, which implies that $\Delta(g') < \Delta(g)$.

Any policy with a minimum service time greater than 1 can be strictly improved by shifting it to start at $s_1=1$. Thus, the optimal policy must have a minimum service time of $s_1=1$.

Now, consider the set of all valid PMFs with $s_1=1$ and leakage constraint $g(1)=\MLC$. Assume there exists a PMF $g(s)$ in this set that does not have the ``greedy" structure. By definition, this means there must exist at least two service times, $s_a$ and $s_b$, with $1 < s_a < s_b$, such that $g(s_a) < \MLC$ and $g(s_b) > 0$. We can construct a new PMF $g'(s)$ by moving a small probability mass $\epsilon > 0$ from the larger service time $s_b$ to the smaller one $s_a$:
\begin{align}
 g'(s_a) &= g(s_a) + \epsilon, \\
 g'(s_b) &= g(s_b) - \epsilon, \\
 g'(s) &= g(s) \quad \text{for } s \notin \{s_a, s_b\}.
\end{align}
For a sufficiently small $\epsilon$, we have $g'(s_a) \le \MLC$ and $g'(s_b) \ge 0$, so $g'(s)$ remains a valid PMF in the feasible set. The leakage constraint is unaffected since $g'(1)=g(1)=\MLC$.

We now show this shifting of probability mass strictly improves the age.
Since $0 < \lambar < 1$ and $s_a < s_b$,
\begin{align}
    \operatorname{E}_{g'}[\lambar^{S-1}] - \operatorname{E}_g[\lambar^{S-1}] &= (g'(s_a)\lambar^{s_a-1} + g'(s_b)\lambar^{s_b-1}) - (g(s_a)\lambar^{s_a-1} + g(s_b)\lambar^{s_b-1}) \\
    &= \epsilon \left(\lambar^{s_a-1} - \lambar^{s_b-1}\right) >0.
\end{align}
Thus $\Delta(g')< \Delta(g)$. Moreover, any PMF that is not of the greedy form can be strictly improved by shifting probability mass from a larger service time to a smaller one. Therefore, the greedy PMF $g_\MLC^*$ uniquely minimizes the average AoI. 

\subsection{Proof of Theorem \ref{thm: RADleakage}}
\label{app:RADleakage}
This proof will employ the following lemma. 
\begin{lemma}
\label{lemma:equivalence}
    For any pair of binary sequences $(y^n, u^n)$, the following equivalence holds:
    \begin{align}
        u^n\in \AU(y^n) \iff y^n \in \AY(u^n).
    \end{align}
\end{lemma}
\begin{proof}
     Let $S_1 = \{(y^n,u^n): u^n\in \AU(y^n)\}$ and $S_2 = \{(y^n, u^n): y^n \in \AY(u^n)\}$. 
The equivalence is shown as:
\begin{align}
    (y^n,u^n) \in S_1 &\iff  u^n \in \AU(y^n)\\
    &\iff (y_i = 1\implies u_i = 1) \text{ for all }i\\
    &\iff (u_i= 0 \implies y_i = 0)\text{ for all }i\\
    &\iff y^n \in \AY(u^n)\\
    &\iff (y^n, u^n) \in S_2.
\end{align}
Thus $S_1 = S_2$.
\end{proof}
    We now apply the result from Lemma~\ref{lemmaRAD} to the definition of MaxL~\eqnref{eqn:maxL_def}.
\begin{align}
     \L(X^n\to Y^n)
&=\log \sum_{y^n\in \support{Y}} P_{u^n}(u^n\in \AU(y^n)) \eqnlabel{definition}\\
&=\log \sum_{y^n\in \support{Y}} \sum_{u^n\in \AU(y^n)}P_{u^n}(u^n) \eqnlabel{eqn:swapsum}\\
&=\log \sum_{u^n} \sum_{y^n\in \AY(u^n)}P_{u^n}(u^n) \eqnlabel{ynsubzn}\\
&=\log \sum_{u^n} \left| \AY(u^n) \right| P_{u^n}(u^n)  \\
&= \log \sum_{u^n} 2^{\mathrm{wt}(u^n)} P_{u^n}(u^n) \\
&= \log \sum_{k=0}^{n} 2^k P_{u^n}\left( \{ \mathrm{wt}(u^n) = k \} \right)\\
&=\log \left(\E{2^{\mathrm{wt}(u^n)}}\right).
\end{align}
To arrive at \eqnref{ynsubzn}, we swap the order of summation. This is justified by Lemma~\ref{lemma:equivalence}.

\subsection{Proof of Lemma~\ref{lemma:unique}}
\label{app:unique}
    Let $f(z) = \E{z^{-D}} = \sum_{d = 1}^ \infty z^{-d} g(d)$. Since $f'(z) = \sum_{d = 1}^ \infty -d z^{-d - 1} g(d) < 0$, for $z> 0$, $f(z)$ is strictly decreasing for positive real numbers. 
    We also have $f(1)= 1$ and $\lim_{z\to\infty}f(z) = 0$. Since $f(z)$ is a continuous and strictly decreasing function that goes from $1$ to $0$ on the interval $[1,\infty]$, it crosses the value $1/2$ exactly once in this interval. This proves the existence and uniqueness of the positive real root $z_0$. 

    For any other root $z$ (real or complex) of $f(z) =\E{z^{-D}}= 1/2$, we have 
    \begin{align}
        |f(z)| = \left|\sum_{d = 1} ^\infty z^{-d} g(d) \right|.
    \end{align}
    Using the triangle inequality we can write
    \begin{align}
        |f(z)| &\le \sum _{d = 1}^\infty {\left|z\right|}^ {-d} g(d)=f(|z|).
    \end{align}
    Since $f(z_0) = \frac{1}{2}$, this implies
    \begin{align}
        f(z_0) = \frac{1}{2} \le f(|z|).
    \end{align}
    Since $f(x)$ is strictly decreasing for $x > 0$, it follows that $z_0 > |z|$. 

\subsection{Proof of Lemma~\ref{thm:RADleakagerate}}
\label{app: RADleakagerate}

    From Lemma \ref{lemma:unique} we know that the unique positive real root $z_0$ of $\E{z^{-D}}=1/2$ has the largest magnitude among all roots of $z_i$. The poles of $M(z)$ are $p_i = 1/z_i$.  
Then using partial fraction expansion, we can express $M(z)$ in terms of the poles $p_0, p_1, \ldots, p_{d_\mathrm{max} - 1}$ as 
\begin{align}
    M(z) &= \frac{C_0}{1 - \frac{z}{p_0}} + \frac{C_1}{1 - \frac{z}{p_1}} + \ldots + \frac{C_{d_\mathrm{max} - 1}}{1 - \frac{z}{p_{d_\mathrm{max} - 1}}}\eqnlabel{mofz},
\end{align}
where $C_0$, $C_1, \ldots , C_{d_\mathrm{max} - 1}$ are constants. Each term in the expansion corresponds to a geometric series. Using the geometric series expansion we have
\begin{align}
    \frac{C_i}{1  - \frac{z}{p_i}} &= C_i \sum_{n = 0}^ \infty {\left(\frac{z}{p_0}\right)}^n \eqnlabel{geoseries}.
\end{align}
Therefore, \eqnref{mofz} and \eqnref{geoseries} imply
\begin{align}
    \sum_{n = 0}^ \infty m(n) z^n &= \sum_{n = 0}^ \infty C_0 {\left(\frac{1}{p_0}\right)}^n z^n + \sum _{n = 0}^ \infty C_1 {\left(\frac{1}{p_1}\right)}^n z^n + \ldots + \sum _{n = 0}^\infty C_{d_\mathrm{max} - 1} {\left(\frac{1}{p_{d_\mathrm{max} - 1}}\right)}^n z^n.
\end{align}
Substituting $p_i= 1/z_i$ for $i \ge 0$ yields
\begin{align}
    m(n) &= C_0 {z_0}^n + C_1 {z_1}^n + \ldots + C_{d_{\max} - 1} {z_{d_\mathrm{max}-1}^n}\\
    &= C_0{z_0}^n \left(1 + \frac{C_1}{C_0} {\left(\frac{z_1}{z_0}\right)}^ n + \frac{C_2}{C_0} {\left(\frac{z_2}{z_0}\right)}^ n + \ldots +  \frac{C_{d_\mathrm{max}-1}}{C_0} {\left(\frac{z_{d_{\mathrm{max}}-1}}{z_0}\right)}^ n \right)\\
    &=  C_0{z_0}^n (1 + \epsilon(n)),
\end{align}
where $\epsilon (n) =  \frac{C_1}{C_0} {\left(\frac{z_1}{z_0}\right)}^ n + \frac{C_2}{C_0} {\left(\frac{z_2}{z_0}\right)}^ n + \ldots +  \frac{C_{d_\mathrm{max}-1}}{C_0} {\left(\frac{z_{d_{\mathrm{max}}-1}}{z_0}\right)}^ n $. Since Lemma \ref{lemma:unique}  guarantees $z_0 > |z_i|$ for all $i \ge 1$, $|z_i/z_0| < 1$.  Thus, $\lim _{n \to \infty} \epsilon(n) = 0$. 

The value of $z_0$ is determined by the  equation $G(1/z_0) = 1/2$, which is equivalent to
\begin{align}
    \E{z_0^{-D}} = \frac{1}{2}.
\end{align}
The MaxL rate is now given by 
\begin{align}
    \Lambda &= \frac{1}{n}\L(X^n\to Y^n) = \lim _ {n\to \infty }\frac{1}{n}\log (m(n))\\
    &= \lim_{n \to \infty} \frac{1}{n} \log(C_0 {z_0}^n (1 + \epsilon(n)))\\
    &= \lim_{n\to\infty} \left( \frac{\log(C_0)}{n} + \log(z_0) + \frac{\log(1 + \epsilon(n))}{n} \right). 
\end{align}
As $n\to\infty$, $\frac{\log(C_0)}{n}\to 0 $ and $\epsilon(n)\to 0$. Thus the MaxL rate is 
\begin{align}
    \Lambda = \log(z_0).
\end{align}
where $z_0$ is the unique positive real root of $\E{z^{-D}} = 1/2$.

\subsection{Proof of Corollary \ref{cor:RADGeo}}
\label{app:RADgeo}
Since the inter-dump attempt times $D$ are i.i.d.~geometric with mean $\E{D}=\tau$, the system is equivalent to a system where dump attempts are i.i.d.~Bernoulli trials with a success probability of $\mu = 1/\tau$. Thus, the total number of attempts $\mathrm{wt}(u^n)$ follows a $\mathrm{Binomial}(n, \mu)$ distribution. The expectation in Theorem~\ref{thm: RADleakage} becomes 
    \begin{align}
    \E{2^{\mathrm{wt}(u^n)}} &= \sum _{k=0} ^n \binom{n}{k} \mu ^k {(1- \mu )}^ {n-k} 2 ^k\\
    &= \sum _{k=0} ^n \binom{n}{k}
    {(2\mu)} ^k {(1- \mu )}^ {n-k} \\
    & = {(2\mu + (1 - \mu))}^n\\
        &= {(1 + \mu)} ^n .
    \end{align}
Thus the MaxL is
\begin{align}
    \L(X^n\to Y^n) &= n \log (1 + \mu )
\end{align}
and the asymptotic leakage rate  is
\begin{align}
    \Lambda = \lim_{n\to\infty} \frac{1}{n} \L(X^n\to Y^n) = \log(1 + \mu).
\end{align}
Substituting $\mu = 1/\tau$ gives the final result
\begin{align}
    \Lambda = \log(1 + 1/\tau).
\end{align}

\subsection{Proof of Corollary \ref{cor:RADunifleakage}}
\label{app:RADunifleakage}

    From Theorem~\ref{thm:RADleakagerate}, the asymptotic MaxL rate is $\Lambda = \log(z_0)$, where $z_0$ is the unique positive real root of the equation $\E{z^{-D}} = 1/2$. The policy is a discrete uniform distribution with mean $\E{D}=\tau$. This corresponds to a distribution over the set $\{1, 2, \ldots, k\}$ where the mean is $\E{D} = (k+1)/2 = \tau$. This gives the relationship $k = 2\tau - 1$. For a discrete uniform distribution on $\{1,2,\ldots, k\}$, the expectation is
   \begin{align}
    \E{z_0^{-D}} &= \sum_{d=1}^{k} g(d) z_0^{-d} \\
    &= \sum_{d=1}^{k} \frac{1}{k} z_0^{-d} \\
    &= \frac{1}{k} \left( \frac{z_0^{-1}(1 - z_0^{-k})}{1 - z_0^{-1}} \right) \\
    %&= \frac{1}{k} \left( \frac{z_0^{-1}(1 - z_0^{-k})}{z_0^{-1}(z_0 - 1)} \right) \\
    &= \frac{1 - z_0^{-k}}{k(z_0 - 1)} \eqnlabel{unifk}.
\end{align}

Substituting $k = 2\tau - 1$ into \eqnref{unifk} and equating \eqnref{unifk} with $1/2$ yields the claim.

\subsection{Justification for the Dinkelbach Method}
\label{app: dinkelbach}
Dinkelbach's method addresses problems of the form $\min_{x \in S} \{N(x)/D(x)\}$, where the following three assumptions hold:
\begin{enumerate}
    \item The set $S$ is a compact and connected subset of a Euclidean space.
    \item The functions $N(x)$ and $D(x)$ are continuous over $S$.
    \item The denominator $D(x) > 0$ for all $x \in S$.
\end{enumerate}
We verify that our optimization problem in \eqnref{eq:OptimalDumpProblem} satisfies these conditions. In our problem, the optimization variable $x$, is the PMF $g = (g(1), g(2), \ldots, g(d_{\max}))$. The numerator is $N(g) = \E{D^2} = \sum_{d = 1}^{d_{\max}} d^2 g(d)$. The denominator is $D(g) = \E{D} = \sum_{d = 1}^{d_{\max}} d g(d)$. The feasible set $S$ is the set of all PMFs $g$ satisfying the constraints in \eqnref{eq:LeakageConstraint} and \eqnref{eq:ProbabilityConstraint}. Now we verify each condition below:
\begin{enumerate}
    \item Compactness and connectedness of $S$: The feasibility set $S$ is the set of all PMFs that satisfy the leakage constraint. Since all of these constraints are linear, the set $S$ is convex, and thus connected. Since the maximum inter-dump time is bounded by $d_\mathrm{max}$, the feasibility set $S$ is a closed and bounded set, and therefore compact.
    \item Continuity of $N(g)$ and $D(g)$: Both the numerator $N(g) = \E{D^2}$ and the denominator $D(g) = \E{D}$ are linear functions with respect to the PMF $g$, and are thus continuous.
    \item Positivity of the denominator $D(g)$: The denominator is $D(g) = \E{D}$. Since the inter-dump time $d$ can only take values in $\{1, \ldots,d_{\max}\}$, the condition $D(g) > 0$ is satisfied for all $g \in S$.
\end{enumerate}

\subsection{Proof of Lemma \ref{lemma:convexity}}
\label{app:convexity}
    Using the chain rule,  $f''(x) = C''(d(x)) \cdot (d'(x))^2 + C'(d(x)) \cdot d''(x)$, we have
    \begin{align}
        f''(x) &= 2 \left(-\frac{1}{x \log(z_0)}\right)^2 + (2d(x) - \gamma) \left(\frac{1}{x^2 \log(z_0)}\right) \\
        &= \frac{1}{x^2 \log(z_0)} \left( \frac{2}{\log(z_0)} - \frac{2\log(x)}{\log (z_0)} - \gamma \right).
    \end{align}
     Since $x = z_0^{-d} > 0$ and $z_0 > 1$, the term $\frac{1}{x^2 \log(z_0)} > 0$. Thus, $f''(x) > 0$ if and only if
    \begin{align}
        \gamma &< \frac{2}{\log(z_0)}(1 - \log (x)).
    \end{align}
     For  $x \in[0, z_0^{-1}]$, $f''(x) > 0$ if  
     \begin{align}
         \gamma & < \frac{2}{\log (z_0)}(1 + \log (z_0))
         = 2 + \frac{2}{\log(z_0)}. \eqnlabel{eq:convexity_condition}
     \end{align}

\subsection{Age Analysis of the Markov Source}
\label{app:markov_age}
We analyze the age $A_s$ of a two-state discrete-time Markov source with an inactive state (0) and an active state (1, corresponding to a generated update). The transition probabilities are $P_{01}$ (inactive to active) and $P_{10}$ (active to inactive), which implies $P_{00} = 1-P_{01}$ and $P_{11} = 1-P_{10}$.

Let $\pi_{i,a}$ be the stationary probability of the joint state where the source is in state $i \in \{0,1\}$ and the age is $A_s=a$. The age $a=1$ occurs only if an update was generated in the previous slot. For an age $a \ge 2$, no update was generated in the previous slot. The balance equations are
\begin{align}
    \pi_{0,a} &= P_{00} \pi_{0, a-1} \quad (a \ge 2), \\
    \pi_{1,a} &= P_{01} \pi_{0, a-1} \quad (a \ge 2).
\end{align}
This implies $\pi_{0,a} = P_{00}^{a-1} \pi_{0,1}$ and $\pi_{1,a} = P_{01}P_{00}^{a-2}\pi_{0,1}$ for $a \ge 2$.

The base cases for age $a=1$ are
\begin{align}
    \pi_{0,1} &= P_{10} \sum_{a=1}^{\infty} \pi_{1,a} = P_{10} \pi_1 ,\\
    \pi_{1,1} &= P_{11} \sum_{a=1}^{\infty} \pi_{1,a} + P_{01} \sum_{a=1}^{\infty} \pi_{0,a} = P_{11}\pi_1 + P_{01}\pi_0,
\end{align}
where $\pi_i = \sum_{a=1}^{\infty} \pi_{i,a}$ are the marginal state probabilities satisfying the balance equation $\pi_0 P_{01} = \pi_1 P_{10}$.

Solving this system yields  $\pi_{0,1} = \frac{P_{01}P_{10}}{P_{01}+P_{10}}$.

The marginal probability distribution of the age $A_s$ is then
\begin{align}
    P(A_s=1) &= \pi_{0,1} + \pi_{1,1} = \frac{P_{01}}{P_{01}+P_{10}}, \\
    P(A_s=a) &= \pi_{0,a} + \pi_{1,a} = (P_{00}^{a-1} + P_{01}P_{00}^{a-2})\pi_{0,1} \nonumber \\
             &= (P_{00}+P_{01})P_{00}^{a-2}\pi_{0,1} = P_{00}^{a-2}\pi_{0,1} \quad (a \ge 2).
\end{align}
The average age at the input to the server $\E{A_s}$ is calculated as follows
\begin{align}
    \E{A_s} &= \sum_{a=1}^{\infty} a P(A_s=a) \\
    &= P(A_s=1) + \sum_{a=2}^{\infty} a P_{00}^{a-2}\pi_{0,1} \\
    &= \frac{P_{01}}{P_{01}+P_{10}} + \frac{\pi_{0,1}}{P_{00}} \sum_{a=2}^{\infty} a P_{00}^{a-1}\\
    &= \frac{P_{01}}{P_{01}+P_{10}} + \frac{P_{10}}{P_{01}+P_{10}}\left(\frac{1-P_{01}}{P_{01}} + 2\right) \\
    &= 1 + \frac{P_{10}}{P_{01}(P_{01}+P_{10})}.
\end{align}

\bibliographystyle{IEEEtran}
\bibliography{privacyagebib.bib,ton.bib}

\end{document}